\documentclass[journal,11pt,onecolumn,draftclsnofoot,]{IEEEtran}

\ifCLASSINFOpdf
\else
   \usepackage[dvips]{graphicx}
\fi
\usepackage{url}

\usepackage{graphicx}
 \usepackage{upgreek}
\usepackage{cite}
\usepackage{amsmath,amsfonts}
\usepackage{algorithmic}
\usepackage{graphicx}
\usepackage{textcomp}
\usepackage{xcolor}
\usepackage{amsthm}
\usepackage{caption}
\usepackage{relsize}
\usepackage{subcaption}
\usepackage{float}
\usepackage{mathbbol}
\def\BibTeX{{\rm B\kern-.05em{\sc i\kern-.025em b}\kern-.08em
    T\kern-.1667em\lower.7ex\hbox{E}\kern-.125emX}}
    
\usepackage{systeme}
\usepackage{color}
\usepackage{pgfplots}
\usepackage{bbm}
    
  \newcommand{\figref}[1]{Fig.~\protect\ref{#1}}

                            \newcommand{\pto}{\overset{P}\longrightarrow }

\long\def\comment#1{}

\DeclareMathOperator*{\argmax}{arg\,max}
\DeclareMathOperator*{\argmin}{arg\,min}

\newfont{\bbb}{msbm10 scaled 700}

\newfont{\bb}{msbm10 scaled 1100}

\newcommand{\ev}{{\bf e}}

\newcommand{\gv}{{\bf g}}

\newcommand{\qv}{{\bf q}}
\newcommand{\rv}{{\bf r}}
\newcommand{\sv}{{\bf s}}

\newcommand{\vv}{{\bf v}}
\newcommand{\xv}{{\bf x}}
\newcommand{\yv}{{\bf y}}
\newcommand{\zv}{{\bf z}}

\newcommand{\Am}{{\bf A}}

\newcommand{\Cm}{{\bf C}}

\newcommand{\Hm}{{\bf H}}
\newcommand{\Id}{{\bf I}}

\newcommand{\Rm}{{\bf R}}
\newcommand{\Sm}{{\bf S}}

\newcommand{\Vm}{{\bf V}}
\newcommand{\Xm}{{\bf X}}

\newcommand{\Sigmam}{\hbox{\boldmath$\Sigma$}}

\newcommand{\Omegam}{\hbox{\boldmath$\Omega$}}

\newtheorem{theorem}{Theorem}

\newtheorem{proposition}{Proposition}
\newtheorem{remark}{Remark}

\newtheorem{assumption}{Assumption}
\newtheorem{corollary}{Corollary}

\begin{document}

\title{Optimum GSSK Transmission in Massive MIMO Systems Using the Box-LASSO Decoder}

\author{Ayed M. Alrashdi, \IEEEmembership{Member, IEEE}, Abdullah E. Alrashdi, \IEEEmembership{Graduate Student Member, IEEE}, Amer Alghadhban, \IEEEmembership{Member, IEEE}, and Mohamed A. H. Eleiwa
\thanks{The authors are with the Department of Electrical Engineering, College of Engineering, University of Ha'il, P.O. Box 2440, Ha'il, 81441, Saudi Arabia.
A. E. Alrashdi is also with Saudi Aramco.}
\thanks{Corresponding author: Ayed M. Alrashdi (e-mail: am.alrashdi@uoh.edu.sa).
}
}

\markboth{IEEE xxx, VOL. xx, NO. x, xxx 2022}
{Shell \MakeLowercase{\textit{Alrashdi et al.}}: IEEE Wireless Communiation Letters}
\maketitle

\begin{abstract}
%
We propose in this work to employ the Box-LASSO, a variation of the popular LASSO method, as a low-complexity decoder in a massive multiple-input multiple-output (MIMO) wireless communication system.
The Box-LASSO is mainly useful for detecting simultaneously structured signals such as signals that are known to be sparse and bounded. 
One modulation technique that generates essentially sparse and bounded constellation points is the so-called generalized space-shift keying (GSSK) modulation. 
%
In this direction, we derive high dimensional sharp characterizations of various performance measures of the Box-LASSO such as the mean square error, probability of support recovery, and the element error rate, under independent and identically distributed (i.i.d.) Gaussian channels that are not perfectly known. 
In particular, the analytical characterizations can be used to demonstrate performance improvements of the Box-LASSO as compared to the widely used standard LASSO. Then, we can use these measures to optimally tune the involved hyper-parameters of Box-LASSO such as the regularization parameter. In addition, we derive optimum power allocation and training duration schemes in a training-based massive MIMO system. Monte Carlo simulations are used to validate these premises and to show the sharpness of the derived analytical results.
\end{abstract}
\begin{IEEEkeywords}
LASSO, box-constraint, performance analysis, channel estimation, spacial modulation, power allocation, massive MIMO.
\end{IEEEkeywords}

\IEEEpeerreviewmaketitle

\section{Introduction}
\label{sec:intro}

\IEEEPARstart{T}{he} least absolute selection and shrinkage operator (LASSO) \cite{tibshirani1996regression} is a widely used method to recover an unknown \textit{sparse} signal $\sv_0$ from noisy linear measurements
$\rv =  \Am \sv_{0} + \vv,$ 
by solving the following optimization problem: 
\begin{equation}
\label{EN_1}   
\widehat{\sv}= { \rm{arg}} \ \underset{{\sv \in \mathbb{R}^{n}}}{\operatorname{\min}} \ \|   \Am \sv - \rv  \|_2^2 + \gamma \| \sv \|_1,
\end{equation}
where $\| \cdot \|_2$, and $ \|\cdot \|_1$ represent the $\ell_2$-norm and $\ell_1$-norm of a vector, respectively. Furthermore, $ \Am$ is the measurement matrix, $\vv$ is the noise vector, and $\gamma>0$ is a regularization parameter that balances between the fidelity of the solution as controlled by the $\ell_2$-norm on one side, and the sparsity of the solution as enforced by the $\ell_1$-norm on the other hand. It allows for learning a sparse model where few of the entries are non-zero.
The LASSO reduces to the linear regression as $\gamma \to 0$.
The LASSO has been widely used in modern data science and signal processing applications such as in \cite{ting2009sparse,gui2012improved, bishop2006pattern}.
%

The LASSO is a special instance of a class of problems called non-smooth regularized convex optimization problems \cite{thrampoulidis2016precise}.
In recent years, various forms of sharp analysis of the asymptotic performance of such optimization problems have been studied under the assumption of noisy independent and identically distributed (i.i.d.) Gaussian measurements. 
They mostly take one of the following approaches.

The first is the approximate message passing (AMP) framework, which was utilized in \cite{donoho2009message, bayati2011dynamics, bayati2012lasso} to conduct a sharp asymptotic study of the performance of compressed sensing problems under the assumptions of i.i.d. Gaussian sensing matrix. 

The authors in \cite{guo2009single, kabashima2010statistical,rangan2012asymptotic} undertook a different approach that uses the replica method from statistical physics, which is a powerful high-dimensional analysis tool. 
However, it lacks rigorous mathematical justifications in some steps.
In addition, the high-dimensional error performance of different regularized estimators
has been previously considered using some heuristic arguments and numerical simulations in \cite{bean2013optimal} and \cite{el2013robust}.  

Another approach based on random matrix theory (RMT) \cite{couillet2011random} was taken in \cite{el2018impact, karoui2013asymptotic} to analyze the high-dimensional squared error performance of ridge regression. One major drawback of RMT is that it requires the involved optimization problems to admit a \emph{closed-form} solution which is not true for the general non-smooth convex optimization problems. 

The most recent approach is based on a newly developed framework that uses the convex Gaussian min-max theorem (CGMT) initiated by Stojnic \cite{stojnic2013framework} and further extended by Thrampoulidis \emph{et al.} in \cite{thrampoulidis2016precise}. It provides the analysis in a more natural and elementary way when compared to the previously discussed methods.
The CGMT has been applied to the performance analyses of various optimization problems. The asymptotic mean square error (MSE) have been analyzed for various regularized estimators in \cite{thrampoulidis2015regularized, thrampoulidis2015lasso, atitallah2017box, thrampoulidis2015asymptotically}.  Precise analysis of general performance measures such as the probability of support recovery and $\ell_1$-reconstruction error of the LASSO and Elastic Net was obtained in \cite{abbasi2016general,alrashdi2017precise,  alrashdi2019precise}.
The asymptotic symbol error rate (SER) of the box relaxation optimization has been derived for various modulation schemes in \cite{alrashdi2020optimum, hayakawa2020asymptotic, thrampoulidis2018symbol, atitallah2017ber}.
CGMT has also been used for the analysis of nonlinear measurements models (e.g., quantized measurements) in \cite{thrampoulidis2015lasso} and \cite{thrampoulidis2018performance}. 
\subsection{Contributions}

In this paper, instead of the standard LASSO in \eqref{EN_1}, we will use the so-called Box-LASSO \cite{atitallah2017box}, which is the same as the LASSO but with an additional box-constraint. 
We will formally define it in \eqref{Box-EN_1}. 
We propose using the Box-LASSO as a low-complexity decoder in massive MIMO communication systems with modern modulation methods such as the generalized space-shift keying (GSSK) \cite{jeganathan2008generalized} modulation and the generalized spatial modulation (GSM) \cite{younis2010generalised}. 
In such systems, the transmitted signal vector $\sv_0$ is inherently sparse and have elements belonging to a finite alphabet, which is an excellent setting for using the Box-LASSO.

Using the CGMT framework, this paper derives sharp asymptotic characterizations of the {mean square error}, {support recovery probability, and element error rate} of the Box-LASSO under the presence of uncertainties in the channel matrix that has i.i.d. Gaussian entries.
In addition, the analysis demonstrate that the Box-LASSO outperforms the standard one in all of the considered metrics.
The derived expressions can be used to optimally tune the involved hyper-parameters of the algorithm.
Furthermore, we study the application of the Box-LASSO to GSSK modulated massive MIMO systems and optimize their power allocation and training duration.

The additional contributions of this paper against the most related works such as \cite{atitallah2017box, hayakawa2020asymptotic, alrashdi2017precise, alrashdi2019precise,bereyhi2020detection} are summarized as follows:\\
$\bullet$ This paper considers the more practical and challenging scenario of \emph{imperfect} channel state information (CSI), whereas \cite{atitallah2017box,hayakawa2020asymptotic,bereyhi2020detection} only derived the analysis for the ideal case of perfect CSI.\\
$\bullet$ Even when the imperfect CSI case was studied in the previous works on the LASSO and Box-Elastic Net in \cite{alrashdi2017precise, alrashdi2019precise}, only a theoretical imperfect CSI model (the so-called Gauss-Markov model) was considered. On the other hand, this work presents the analysis under a more practical model of the imperfect CSI, which is the linear minimum mean square error (LMMSE) channel estimate in \eqref{eq:LMMSE} for a massive MIMO application.\\
$\bullet$ With this massive MIMO application in mind, we derive the asymptotically optimal power allocation and training duration schemes for GSSK signal recovery.\\
$\bullet$ We show that the derived power allocation optimization is nothing but the well-known scheme that maximizes the effective signal-to-noise ratio (SNR) in \cite{hassibi2003much}.
%

\subsection{Organization}
The rest of this paper is organized as follows. The system model, channel estimation and the proposed Box-LASSO decoder are discussed in Section~\ref{sec:set-up}.
Section~\ref{sec:Results} provides the main asymptotic results of the work. The application of Box-LASSO to a massive MIMO system and its power allocation and training duration optimization is presented in Section~\ref{sec:GSSK}. 
Finally, concluding remarks and future research directions are stated in Section~\ref{sec:conc}.
The proof of the main results is deferred to the Appendix.
\subsection{Notations}
Here, we gather the basic notations that are used throughout the paper. 
Bold face lower case letters (e.g., $\xv$) represent a column vector while $x_j$ is its $j^{th}$ entry.
For $\xv \in \mathbb{R}^n$, let $\| \xv \|_2 = \sqrt{\sum_{j=1}^n x_j^2}$, and $\|\xv \|_1= \sum_{j=1}^n |x_j|$.
Matrices are denoted by upper case letters such as $\Xm$, with $\Id_n$ being the $n \times n$ identity matrix.
$(\cdot)^T$ and $(\cdot)^{-1}$ are the transpose and inverse operators respectively. 
We use the standard notations $\mathbb{E}[\cdot],$ and $\mathbb{P}[\cdot]$ to denote the expectation of a random variable and probability of an event respectively. We write $X \sim p_x$ to denote that a random variable $X$ has a probability density/mass function $p_x$. In particular, the notation $\qv \sim \mathcal{N}(\mathbf{0},\Rm_{\qv})$ is used to denote that the random vector $\qv$ is normally distributed with $\mathbf{0}$ mean and covariance matrix $\Rm_{\qv} = \mathbb{E}[\qv \qv^T]$, where $\mathbf{0}$ represent the all-zeros vector of the appropriate size. 
The notation $\delta_\xi$ is used to represent a point-mass distribution at $\xi$.
We write $``\pto "$ to denote convergence in probability as $n \to \infty$.  We also use standard notation ${\rm{plim}}_{n\to\infty} \Theta_n = \Theta$ to denote that a sequence of random variables
$\Theta_{n},[n=1,2,...]$, converges in probability towards a constant $\Theta$. When writing $x^\star = \mathrm {arg} \min_x f(x)$, the operator $\mathrm{arg} \min$ returns any one of the possible minimizers of $f$. 
The function $Q(x) = \frac{1}{\sqrt{2 \pi}} \int_{x}^{\infty} e^{-u^2/2} {\rm d}u$ is the $Q$-function associated with the standard normal density.

\noindent
Finally,  for $a,\gamma,l,u \in \mathbb{R}$, such that $\gamma,u \geq 0, l \leq 0, $ we define the following functions:\\
$\bullet$ The \textit{saturated shrinkage} function $\mathcal{H}(a ; \gamma, l, u) :=\argmin_{l \leq s \leq u}$ $\frac{1}{2}(s-a)^2 + \gamma |s|$, which is given as: 
\begin{equation}\label{soft_TH}
\mathcal{H}(a ; \gamma, l, u) =
\begin{cases} 
         u & ,\text{if}  \ a \geq u + \gamma \\ 
           
         a - \gamma & ,\text{if}  \ \gamma < a < u + \gamma \\      
   
		 0 & ,\text{if} \  |a| \leq \gamma  \\
			
		 a + \gamma & ,\text{if}  \ l - \gamma < a < - \gamma \\      

         l  & ,\text{if} \  a \leq l - \gamma.
\end{cases}
\end{equation}
A plot of this function is depicted in \figref{fig:Soft}.\\
\begin{figure}
\begin{center}
\includegraphics[width=6.40cm, height = 4.5cm]{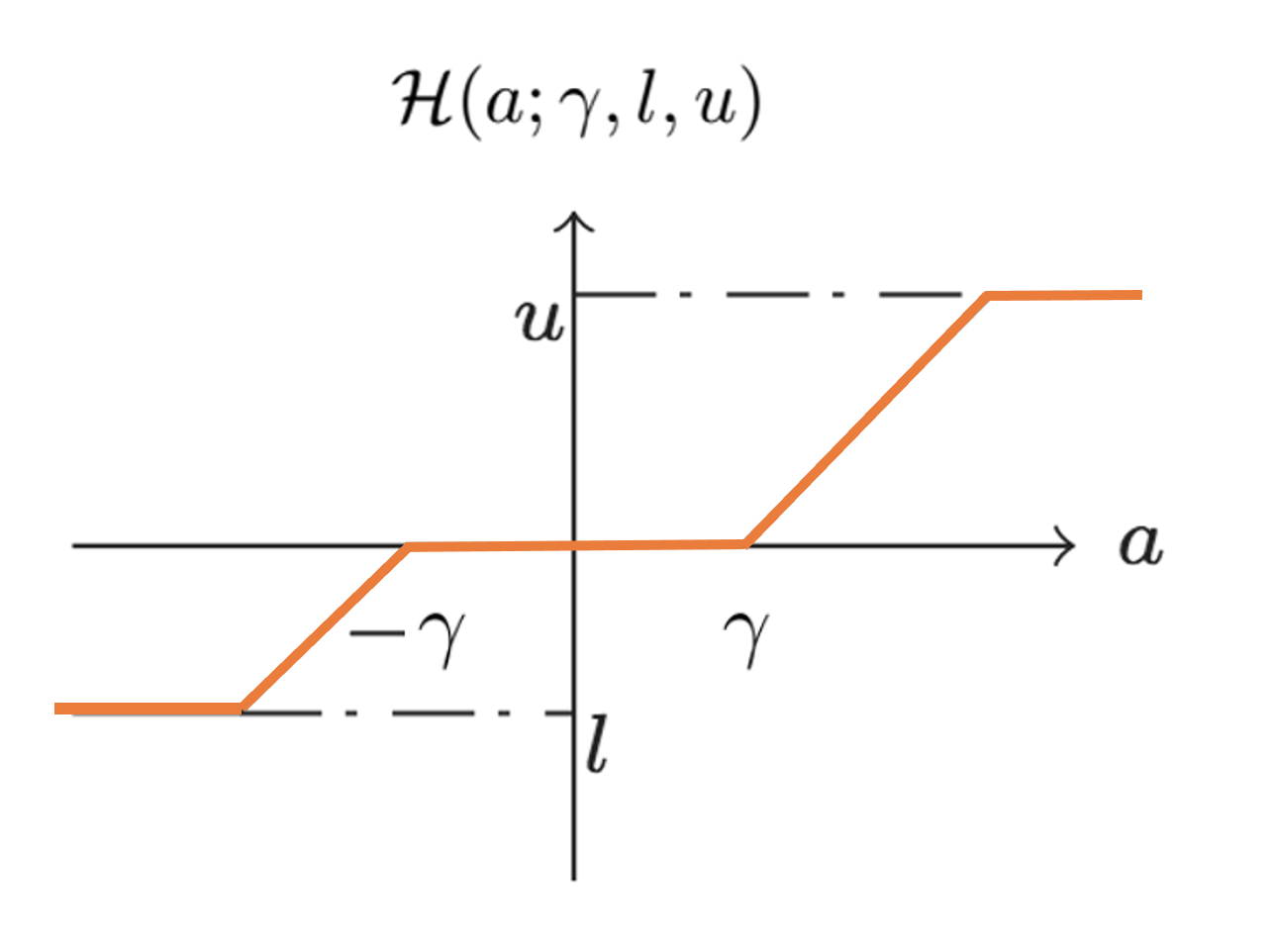}
\caption{\scriptsize{Saturated shrinkage function.}}%
\label{fig:Soft}
\end{center}
\end{figure}
$\bullet$ Also, let $\mathcal{J}(a;\gamma,l,u) :=\min_{l \leq s \leq u} \frac{1}{2}(s-a)^2 + \gamma |s|$, which can be rewritten as:
\begin{equation}\label{optimal_Threshold}
\mathcal{J}(a ; \gamma, l, u) =
\begin{cases} 
      \frac{1}{2} (u - a)^2 + \gamma u  & ,\text{if}  \ a \geq u + \gamma \\      
					
	   \gamma a - \frac{1}{2} \gamma^2  & ,\text{if}  \ \gamma < a < u + \gamma \\     
	    
			\frac{1}{2} a^2 & ,\text{if} \  |a| \leq \gamma   \\
			
      -\gamma a - \frac{1}{2} \gamma^2  & ,\text{if} \  l - \gamma < a < -\gamma \\
      
            \frac{1}{2} (l - a)^2 - \gamma l & ,\text{if} \  a \leq l - \gamma. 
\end{cases}
\end{equation}

\section{Problem Setup}
\label{sec:set-up}
\subsection{System Model}
\label{w_A}
A massive MIMO system with $n$ transmit (Tx) antennas and $m$ receive (Rx) antennas is considered in this paper. 
We herein consider a training-based transmission that consists of a coherence interval with $T = T_t + T_d$ symbols in which the channel realization is assumed to be constant.  During the first part of this coherence interval, $T_t$ symbol intervals are used as known pilot symbols, with average power, ${P}_t $.  These pilot symbols are employed for channel estimation purposes. The remaining $T_d$ symbols are dedicated to transmitting data symbols with average power, $P_d$. 
Conservation of energy implies that
\begin{equation}
\label{eq:energy conseve}
P_t  T_t + P_d T_d =   {P} \  T,
\end{equation}
where $P$ is the average total transmission power.

Letting $\nu$ denote the ratio of the total transmission energy allocated to the data, we may write
\begin{align}
P_d T_d= \nu T  P, \  \ \quad P_t  T_t =  (1- \nu)  T {P}, \ \quad \nu \in (0,1).
\end{align}
This system model is illustrated in \figref{MIMO:fig}.

The {received} signal for the \emph{data} transmission phase can be modeled as
\begin{equation}
\rv = \sqrt{\frac{P_d}{n}} \Hm \sv_0 +\vv \in \mathbb{R}^m,
\end{equation}
where the following model-assumptions hold, except if otherwise stated:
\begin{itemize}
\item $\Hm \in \mathbb{R}^{m \times n}$ is the MIMO channel matrix which has i.i.d. standard Gaussian entries (i.e., $\mathcal{N}(0,1)$).
\item $\vv \in \mathbb{R}^{m}$ is the noise vector with i.i.d. standard Gaussian entries, i.e., $\vv \sim \mathcal{N}(\mathbf{0},\Id_m)$.
\item The unknown signal vector $\sv_{0}$ is assumed to be $k$-sparse, i.e., only $k$ of its entries are sampled i.i.d. from a distribution $ p_{s_0}$ which has zero-mean and unit-variance (i.e., $\mathbb{E}[S_0^2] = \sigma_s^2=1$), and the remaining entries are zeros. 
\end{itemize}
\begin{figure}[ht]
  \centering
\includegraphics[width = 11.5cm]{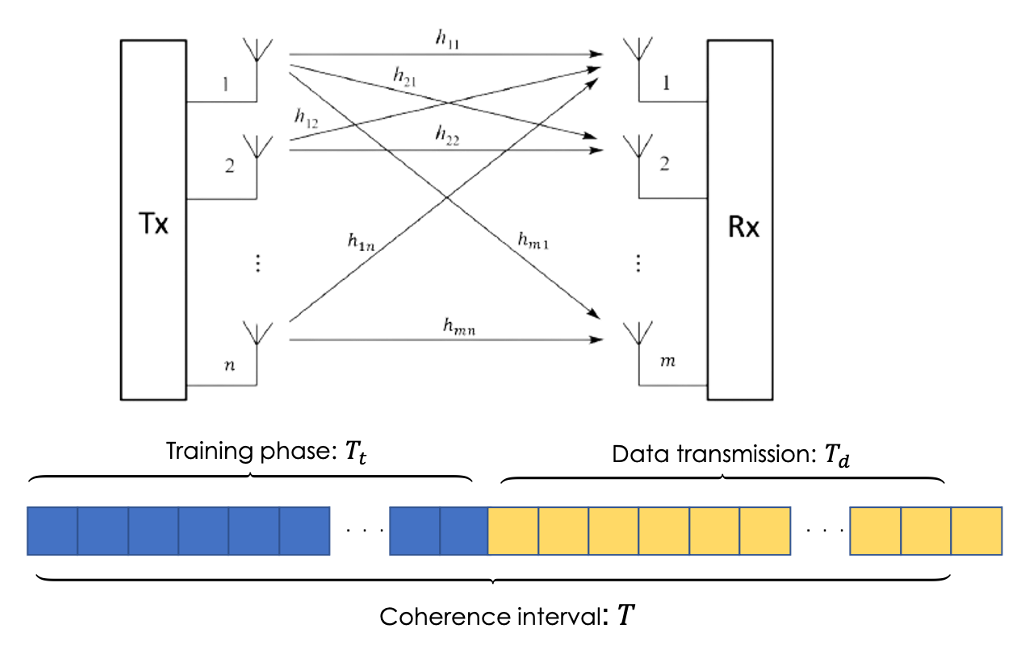}
\caption{\scriptsize{A training-assisted massive MIMO system.}}
\label{MIMO:fig}
\end{figure}
\subsection{Estimation of the Channel Matrix}
As indicated in the preceding subsection, a {training phase} in which the transmitter sends $T_{\mathsmaller{t}}$ {pilot} symbols is used to estimate the channel matrix $\Hm$,  which is unknown to the receiver.\footnote{In communications literature, this is known as the ``imperfect CSI'' case.}
In this \emph{training phase}, the received signal is represented as
\begin{equation}
\label{eq:pilot model}
\Rm_\mathsmaller{t} = \sqrt{\frac{{P}_\mathsmaller{t}}{n}}  \Hm \Sm_\mathsmaller{t} + \Vm_\mathsmaller{t},
\end{equation}
where $\Rm_\mathsmaller{t}\in \mathbb{R}^{m \times T_\mathsmaller{t}}$ is the received signal matrix, $\Sm_\mathsmaller{t}\in \mathbb{R}^{n \times T_\mathsmaller{t}}$ is the matrix of transmitted pilot symbols, and $\Vm_\mathsmaller{t}\in \mathbb{R}^{m \times T_\mathsmaller{t}}$ is a zero-mean additive white Gaussian noise (AWGN) matrix with covariance $\mathbb{E}[\Vm_{\mathsmaller{t}} \Vm_{\mathsmaller{t}}^T] =T_{\mathsmaller{t}} \Id_m$.

In this paper, we consider the linear minimum mean square error (LMMSE) estimate of the channel matrix, which can be derived based on the knowledge of $\Rm_\mathsmaller{t}$ from \eqref{eq:pilot model} as \cite{kay1993fundamentals}
\begin{align}
\label{eq:LMMSE}
\widehat{\Hm} &= \sqrt{\frac{n}{{P}_\mathsmaller{t}}}\Rm_\mathsmaller{t} \Sm_{\mathsmaller{t}}^T \left( \Sm_{\mathsmaller{t}} \Sm_{\mathsmaller{t}}^T + \frac{n}{{P}_\mathsmaller{t}} \Id_\mathsmaller{n} \right)^{-1},\nonumber \\
& = \Hm + \Omegam,
\end{align}
where $\Omegam$ is the channel estimation \textit{error matrix}, which is uncorrelated with $\widehat{\Hm}$, as per the orthogonality principle of the LMMSE estimation \cite{kay1993fundamentals, hassibi2003much}. For i.i.d. MIMO channels, it has been proven that the optimal $\Sm_\mathsmaller{t}$ that minimizes the estimation MSE satisfies \cite{hassibi2003much}
\begin{equation}
\label{eq:pilot orthogonality}
 \Sm_\mathsmaller{t} \Sm_{\mathsmaller{t}}^{T}= T_\mathsmaller{t}\Id_\mathsmaller{n}.
\end{equation}
For the above condition to hold,  it is required that 
\begin{align}\label{eq:n pilots}
T_t \geq n. 
\end{align}
Moreover, under \eqref{eq:pilot orthogonality},  it can be shown that the channel estimate $\widehat{\Hm}$ has i.i.d. zero-mean Gaussian entries with variance $\sigma_{{\widehat{H}}}^2 = 1-\sigma_{{\omega}}^2$ \cite{hassibi2003much}, where
\begin{equation}
\label{eq:error matrix variance}
\sigma_{{\omega}}^2 = \frac{1}{ 1+ \frac{{P}_\mathsmaller{t}}{n} T_\mathsmaller{t} }
\end{equation}
is the variance of each element in $\Omegam$. Furthermore, it can be shown that the entries of $\Omegam$ are i.i.d. $\mathcal{N}(0,\sigma_\omega^2)$ distributed.
Note that the training-phase energy $ {P}_\mathsmaller{t} T_\mathsmaller{t}$ controls the quality of the estimation as it appears from \eqref{eq:error matrix variance}. 
In fact, as ${P}_\mathsmaller{t} T_\mathsmaller{t} \to \infty$, $\sigma_\omega^2 \to 0$, and $\widehat{\Hm} \to \Hm$, which represents the case of perfect CSI.
\subsection{Data Detection via the Box-LASSO}
In this work, the problem in (\ref{EN_1}) is referred to as the \textit{standard} LASSO, and we instead introduce the following revised formulation of it termed the \textit{Box-LASSO}:
\begin{equation}\label{Box-EN_1} 
\widehat{\sv}=  {\rm{arg}} \ \underset{{\sv \in \mathbb{B}^n}}{\operatorname{\min}} \bigg\|   \sqrt{ \frac{P_d}{n}} \widehat{\Hm} \sv - \rv  \bigg\|_2^2 +  \gamma P_d  \| \sv \|_1,
\end{equation} 
                  $$ \text{where}, \mathbb{B} = [\ell, \mu],   \ \text{and} \ \ell \leq 0, \mu \geq 0 \in \mathbb{R}.$$
When compared to (\ref{EN_1}), we use $\Am = \sqrt{ \frac{P_d}{n}} \widehat{\Hm} $ here. This is due to the fact that $\Hm$ is not perfectly known and we only have its estimate $\widehat{\Hm}$ that was obtained by training. In addition,  note that the regularization parameter $\gamma$ is scaled by a factor of $P_d $. This is made such that the two terms grow with the same pace.

The only difference between (\ref{Box-EN_1}) and (\ref{EN_1}), is that (\ref{Box-EN_1}) now has a ``box-constraint''. However, as we will show later, in cases where the elements of $\sv_0$ are bounded or approximately so, this minor modification ensures a considerable gain in performance. Therefore, the Box-LASSO can be used to recover simultaneously structured signals \cite{oymak2015simultaneously}, for example, signals that are both bounded and sparse.
Such signals appear in various applications including machine learning \cite{bishop2006pattern}, wireless communications \cite{jeganathan2009space}, image processing \cite{ting2009sparse}, and so on.
Although the Box-LASSO is not as well-known as the standard LASSO, there are a few references where it has been applied \cite{gaines2018algorithms,james2012constrained, stojnic2010recovery}. Of particular interest is the application of the Box-LASSO in spatially modulated MIMO systems such as GSSK modulated signals which we will discuss in Section~\ref{sec:GSSK}.

\subsection{Technical Assumptions}
\label{Sec:assumptions}
In this work, we require the following technical assumptions to hold.
\begin{assumption}\normalfont
The analysis requires that the system dimensions ($m$, $n$ and $k$) grow simultaneously large (i.e., $m,n,k \to \infty$) at fixed ratios:
$$\frac{m}{n} \to \eta \in (0,\infty),$$ and $$\frac{k}{n} \to \kappa \in (0,1).$$
\end{assumption}
\begin{assumption}\normalfont
We assume that the normalized coherence interval, normalized number of pilot symbols and normalized
number data symbols are fixed and given as
$$\frac{T}{n} \to \tau \in (1,\infty),$$ 
$$ \frac{T_{\mathsmaller{t}}}{n}  \to \tau_{\mathsmaller{t}} \in [1,\infty),$$ and $$\frac{T_{\mathsmaller{d}}}{n} \to \tau_{\mathsmaller{d}},$$
respectively.  
\end{assumption}
Under these assumptions, the energy conservation in \eqref{eq:energy conseve} becomes
\begin{equation}
P_t  \tau_t + P_d \tau_d =   {P} \  \tau,
\end{equation}
and the channel estimation error variance in \eqref{eq:error matrix variance} reads
\begin{equation}
\sigma_{{\omega}}^2 = \frac{1}{ 1+ P_t\tau_\mathsmaller{t} }.
\end{equation}
\subsection{Figures of Merit}
We measure the performance of the Box-LASSO using following figures of merit:\\
\textbf{Mean Square Error}: A widely used figure of merit is the \emph{estimation} mean square error (MSE), that measures the divergence of the estimate $\widehat{\sv}$ from the original signal $\sv_0$. Formally, it is defined as 
\begin{align}
{\rm{MSE}} := \frac{1}{n}\| \widehat{\sv} - \sv_0 \|_2^2.
\end{align}
\textbf{Support Recovery}: In sparse recovery problems, a natural performance measure that is employed in numerous applications is \emph{support recovery}, that can be defined as determining whether an element of $\sv_0$ is non-zero (i.e., on the support), or if it is zero (i.e., off the support). The decision, based on the Box-LASSO solution $\widehat{\sv}$, proceeds as follows: if $| \widehat{s}_{j}| \geq \upzeta$, then, $\widehat{s}_j$ is on the support, where $\upzeta > 0$ is a user-defined hard threshold on the elements of $\widehat{\sv}$. Otherwise, $\widehat{s}_j$ is off the support. 

Essentially, we have two measures: the probability of successful \emph{on}-support recovery denoted by $\Uppsi_{\upzeta}^{\text{on}}(\widehat{\sv})$, and the probability of successful \emph{off}-support recovery, i.e., $\Uppsi_{\upzeta}^{\text{off}}(\widehat{\sv})$. Formally, these quantities are defined, respectively, as
\begin{subequations}\label{supp}
\begin{align}
\Uppsi_{\upzeta}^{\text{on}}(\widehat{\sv}) := \frac{1}{k} \sum_{j \in \mathcal{T}(\sv_0)} \mathbbm{1}_{\{| \widehat{s}_{j}| \geq \upzeta \}},\\
\Uppsi_{\upzeta}^{\text{off}}(\widehat{\sv}) := \frac{1}{n-k} \sum_{j \notin \mathcal{T}(\sv_0)} \mathbbm{1}_{\{| \widehat{s}_{j}| \leq \upzeta \}},
\end{align}
\end{subequations}
\noindent
where $\mathbbm{1}_{\{ . \}}$ is the indicator function, and $\mathcal{T}(\sv_0)$ is the support of $\sv_0$, i.e., the set of all non-zero elements of $\sv_0$.
\section{Main Results}
\label{sec:Results}
\subsection{Performance Characterization}
In this subsection,  we summarize the main theoretical results regarding the asymptotic performance of the Box-LASSO.
The first theorem gives the sharp performance analysis of the MSE of the Box-LASSO.
\begin{theorem}\label{EN_mse}\normalfont
Let $\widehat{\sv}$ be a minimizer of the Box-LASSO problem in (\ref{Box-EN_1}), where $\Hm, \vv$ and $\sv_0$ satisfy the model assumptions in Section \ref{w_A}.
In addition, assume that the optimization problem: $\max_{\beta > 0}  \min_{\lambda>0} \ \mathcal{G} (\beta, \lambda)$ has a unique optimizer $(\beta_\star, \lambda_\star)$, where
\begin{align}\label{eq:scalar1}
 \mathcal{G} (\beta, \lambda)&:= \frac{\beta \sqrt{\eta}}{2 \lambda }+ \frac{\beta \lambda \sqrt{\eta}}{2 } \left( 1+ \kappa   {P_d  \sigma_{\omega}^2}  \right) -\frac{\beta^2}{4}   - \frac{\beta}{2 \lambda \sqrt{\eta}} \nonumber \\ 
& +\beta \lambda \sqrt{\eta} P_d  \sigma_{\widehat{H}}^2 \ \mathbb{E} \left[\mathcal{J} \left(S_{0}+\frac{Z}{\lambda \sqrt{\eta P_d  \sigma_{\widehat{H}}^2 }}; \frac{\gamma}{\beta \lambda \sqrt{\eta}  \sigma_{\widehat{H}}^2 }, \ell,\mu \right)   \right],
\end{align}
and the expectation is taken over $S_0 \sim p_{s_0}$ and $Z \sim \mathcal{N}(0,1)$.

Then, under Assumption~1 and Assumption~2,  and for a fixed $\gamma> 0$, it holds:
\begin{align}\label{eq:MSE}
&\underset{n\to\infty}{\rm{plim}} \ {\rm{MSE}}=  \frac{1}{P_d  \sigma_{\widehat{H}}^2} \left(  \frac{1}{\lambda_\star^2} -1 - \kappa   P_d  \sigma_{\omega}^2  \right).
\end{align}
\end{theorem}
\begin{proof}
The proof is relegated to the appendix.
\end{proof}
\begin{remark}[Finding optimal scalars]\normalfont
The scalars $\beta_\star$ and $\lambda_\star$ can be numerically evaluated by solving the first-order optimality conditions, i.e., 
\begin{align}
\nabla_{(\beta, \lambda)}  \mathcal{G}(\beta, \lambda )=\bf {0}.
\end{align}
\end{remark}
\begin{remark}[Roles of $\lambda_\star$ and $\beta_\star$]\normalfont 
\label{Remark:Res}
From Theorem~\ref{EN_mse} above,  we can see that the optimal scalar $\lambda_\star$  is related to the asymptotic MSE.
However, the role of $\beta_\star$ is not evident from the above theorem. 

Based on our derivations, it turns out that $\beta_\star$ is related to another performance metric called the \emph{residual} \cite{hayakawa2020noise} between $\rv$ and the estimate $\widehat{\rv}$ which is also called the \textit{prediction error}. Formally, it is defined as
\begin{align}
{\rm{\mathcal{R}}} : =\frac{1}{n} \bigg\|\underbrace{ \sqrt{\frac{P_d }{n}} \widehat{\Hm}\widehat{\sv}}_{:=\widehat{\rv}} - \rv \bigg\|_2^2.
\end{align}
Then, as we will prove in the Appendix, under the same assumptions in Theorem~\ref{EN_mse}, it holds 
\begin{align}\label{eq:res}
&\underset{n\to\infty}{\rm{plim}} \ \mathcal{R}=  \frac{1}{4} \beta_\star^2.
\end{align}
The above expression clearly shows the role $\beta_\star$ in predicting the asymptotic value of the residual, and 
\figref{R_Fig} illustrates its high accuracy.

The residual metric is less relevant to the MIMO application considered in this paper, but it is of great interest in other practical data science problems, where you only have an access to the measurement vector $\rv$, and not the true vector $\sv_0$.
\end{remark}
\begin{remark}[Optimal MSE regularizer]\normalfont
Theorem~1 can be used to find the optimal regularizer $\gamma_\star^{\mathrm{MSE}}$ that minimizes the MSE. See for example \figref{Fig:mse/res}a.
Particularly,  $\gamma_\star^{\mathrm{MSE}}$ can be found as follows
\begin{align}\label{opt_Reg_MSE}
\gamma_\star^{\mathrm{MSE}} = &\argmin_{\gamma>0} \frac{1}{\lambda_\star},\nonumber \\ 
= &\argmax_{\gamma>0} {\lambda_\star}.
\end{align}
The above expression can be easily proven, by noting that $\gamma$ appears in the MSE expression of \eqref{eq:MSE} only implicitly through $\lambda_\star$. 
\end{remark}

In the next theorem, we sharply characterizes the support recovery metrics introduced earlier in (\ref{supp}).
\begin{theorem}\label{EN_on/off}\normalfont
Under the same settings of Theorem \ref{EN_mse}, for any fixed $\upzeta>0$,
and under Assumption~1 and Assumption~2, the on-support and off-support probabilities converge as
\begin{equation}\label{eq:on_supp}
\underset{n\to\infty}{\rm{plim}}\ \Uppsi_{\upzeta}^{\rm{on}}(\widehat{\sv}) = \mathbb{P} \biggl[\biggl | \mathcal{H}\biggl( S_{0}+\frac{Z}{\lambda_\star \sqrt{\eta P_d  \sigma_{\widehat{H}}^2 }}; \frac{\gamma }{\beta_\star \lambda_\star \sqrt{\eta}   \sigma_{\widehat{H}}^2 }, \ell,\mu  \biggr)  \biggr |   \geq \upzeta \biggr], 
\end{equation}
and 
\begin{equation}\label{eq:off_supp}
\underset{n\to\infty}{\rm{plim}} \ \Uppsi_{\upzeta}^{\rm{off}}(\widehat{\sv}) = \mathbb{P} \biggl[ \biggl| \mathcal{H}\biggl( \frac{Z}{\lambda_\star \sqrt{\eta P_d  \sigma_{\widehat{H}}^2 }}; \frac{\gamma }{\beta_\star \lambda_\star \sqrt{\eta}   \sigma_{\widehat{H}}^2 }, \ell,\mu  \biggr)  \biggr|   \leq \upzeta \biggr],
\end{equation}
respectively.
\end{theorem}
\begin{proof}
The proof of Theorem~\ref{EN_on/off} to a great extent follows the proof of Theorem \ref{EN_mse}, but is omitted for briefness of the presentation. See \cite{thrampoulidis2018symbol,alrashdi2020precise} for similar proof techniques.
\end{proof}
\begin{remark}[Regularizer's strength]\normalfont
%
It is easy to see from Theorem~\ref{EN_on/off} that when $\gamma$ grows larger, $\Uppsi_{\upzeta}^{\rm{off}}(\widehat{\sv})$ converges to $1$ whereas $ \Uppsi_{\upzeta}^{\rm{on}}(\widehat{\sv})$ converges to $0$. When $\gamma$ approaches $0$, opposite behavior is exhibited.
This is expected since large values of $\gamma$ emphasize the $\ell_1$-norm term, resulting in a sparser solution. This is illustrated clearly in \figref{fig:on/off}.
\end{remark}
\begin{remark}[Optimal regularizer]\normalfont
A sensible measure of performance to trade-off between the on-support and off-support recovery probability is
\begin{align}
\Uppsi_{\upzeta}(\widehat{\sv}) :=\theta  \ \Uppsi_{\upzeta}^{\rm{on}}(\widehat{\sv})+ (1-\theta) \ \Uppsi_{\upzeta}^{\rm{off}}(\widehat{\sv}), \ \text{for} \ \theta \in [0, 1].
\end{align}
The behavior of this metric as a function of $\gamma$ is sharply characterized by Theorem~\ref{EN_on/off}. As a result, Theorem~\ref{EN_on/off} can also be utilized to determine the optimal value of $\gamma$ which optimizes $\Uppsi_{\upzeta}(\widehat{\sv})$.
\end{remark}
\subsection{Numerical Illustrations}
%
For the sake of illustration, we will simply look at the instance when $\sv_0$ has elements that are only allowed to take one of two possible values: $0$, or $E > 0$. 
For a normalized sparsity level $\kappa \in (0, 1)$, such prior knowledge is typically modeled using a sparse-Bernoulli
distribution on the elements of $\sv_0$, i.e., $S_0 \sim (1-\kappa)\delta_0 + \kappa \delta_E $.
This model is frequently seen in MIMO communication systems using generalized space-shift keying (GSSK) modulation \cite{jeganathan2008generalized}; we go over the role of the Box-LASSO in such systems in Section~\ref{sec:GSSK}.
In this case, setting $\ell=0$, and $\mu=E$ as the box-constraint values is a natural choice.

Therefore, in our numerical simulations, we assume that $\sv_0$ has elements that are sampled i.i.d. from a sparse-Bernoulli distribution with $\mathbb{P}[S_0 = 0] =0.8$, $\mathbb{P}[S_0 = 1] =0.2$ (i.e., $\kappa=0.2$) and $E =1$; to satisfy the unit-variance assumption. 

\figref{Fig:mse/res} shows the close match between Theorem~\ref{EN_mse} asymptotic prediction of the MSE and residual of the Box-LASSO and the Monte Carlo (MC) simulations. For the simulations, we used $\eta = 1.5,n = 100$, $T = 500,T_t =n, \nu = 0.5,$ and $P= 15 \ {\rm{dB}}$. These results are averaged over $100$ independent realizations
of $\sv_0, \Hm$ and $\vv$.
From this figure, it can be seen that the Box-LASSO outperforms the standard LASSO.
We can also see in Fig. \ref{mse_Fig} that as the regularization parameter $\gamma$ is varied, a pronounced minimum for a certain $\gamma> 0$ is observed. 
\begin{figure*}
\begin{subfigure}{.5\textwidth}
  \centering
  \includegraphics[width=8cm, height =6.1cm]{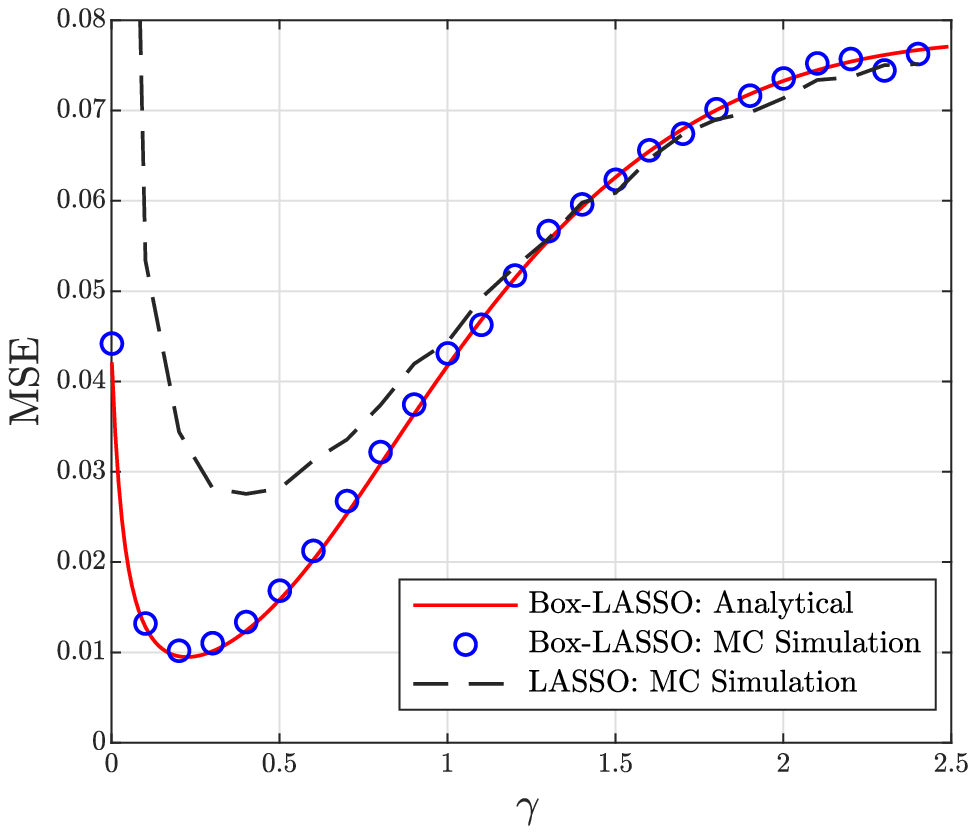}%
\caption{\scriptsize{MSE Performance.}}%
\label{mse_Fig}
\end{subfigure}
\begin{subfigure}{.5\textwidth}
\centering
\includegraphics[width=8cm, height =6.1cm]{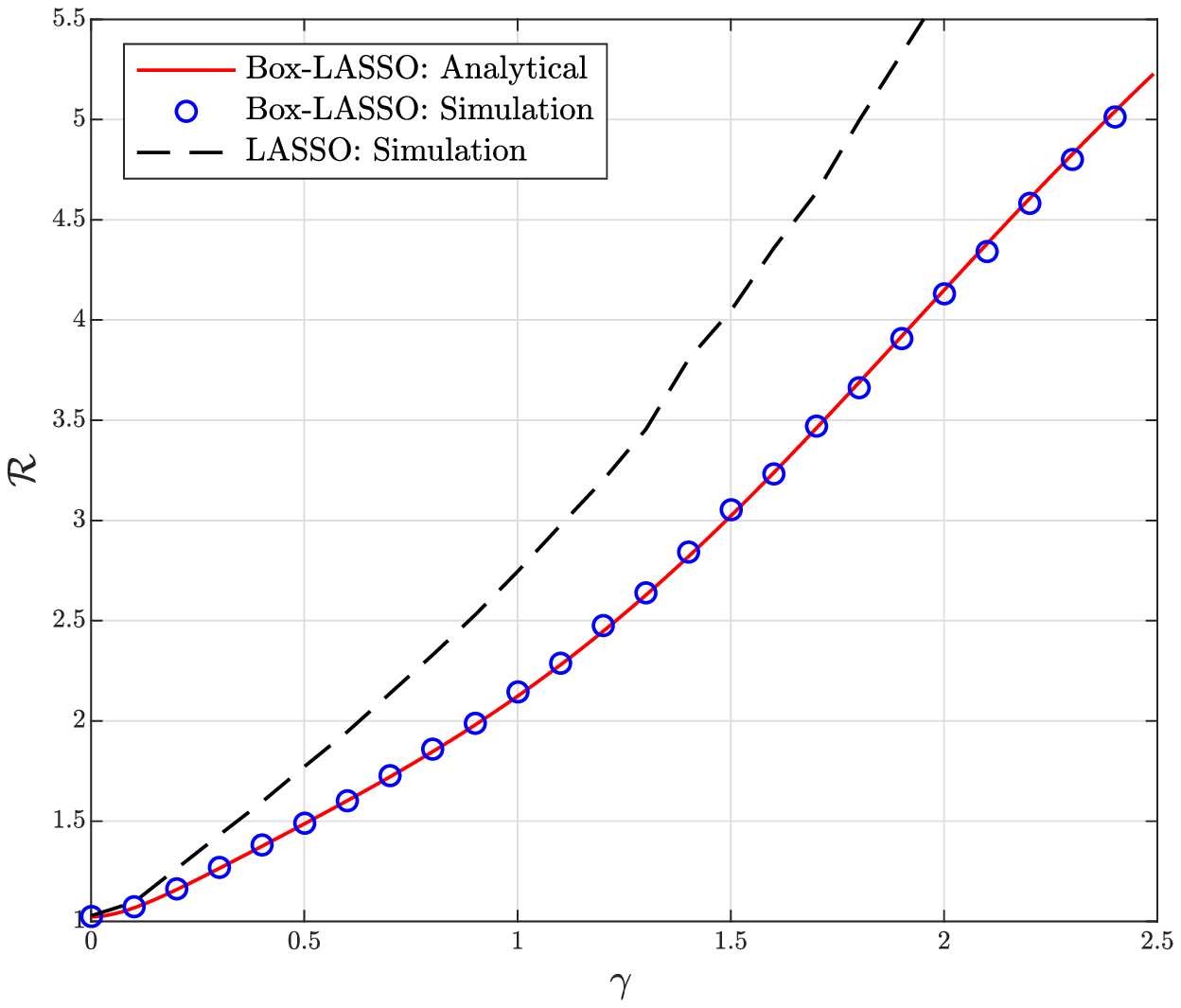}
\caption{\scriptsize{Residual Performance.}}
\label{R_Fig}
\end{subfigure}
\caption{\scriptsize{MSE/Residual performance of the Box-LASSO and the standard LASSO vs. the regularizer. 
The analytical prediction is based on Theorem \ref{EN_mse} with $p_{s_0} = (1-\kappa) \ \delta_0 + \kappa \ \delta_E$. We used $\kappa =0.2,\eta = 1.5,n = 128$, $T = 500,T_t =n, \nu = 0.5, E=1,$ and $P= 15 \ {\rm{dB}}$.}}
\label{Fig:mse/res}
\end{figure*}

The analytical expressions of Theorem~\ref{EN_on/off} for the support recovery probabilities are compared to the MC simulations and displayed in \figref{fig:on/off} with the same simulation settings as in the preceding figure. 
Once again, this figure demonstrates the great accuracy of the provided theoretical expressions.
\begin{figure*}
\begin{subfigure}{.5\textwidth}
  \centering
\includegraphics[width=8cm, height =6.1cm]{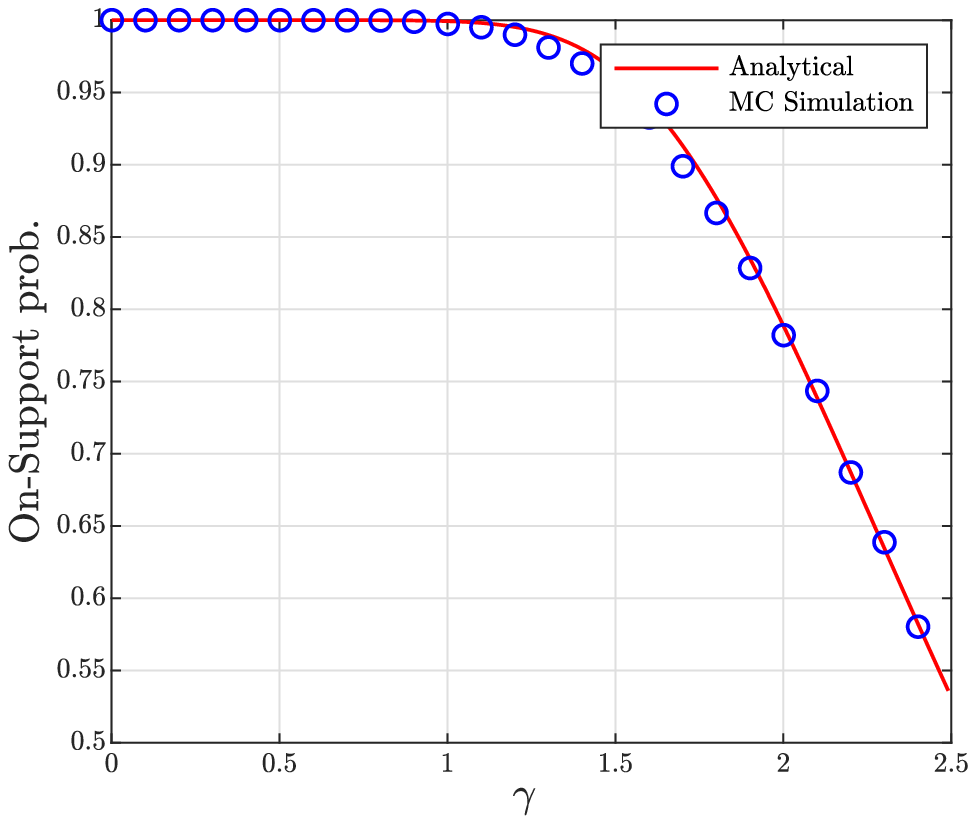}
\caption{\scriptsize{On-Support Probability.}}
\end{subfigure}
\begin{subfigure}{.5\textwidth}
\centering
\includegraphics[width=8cm, height =6.1cm]{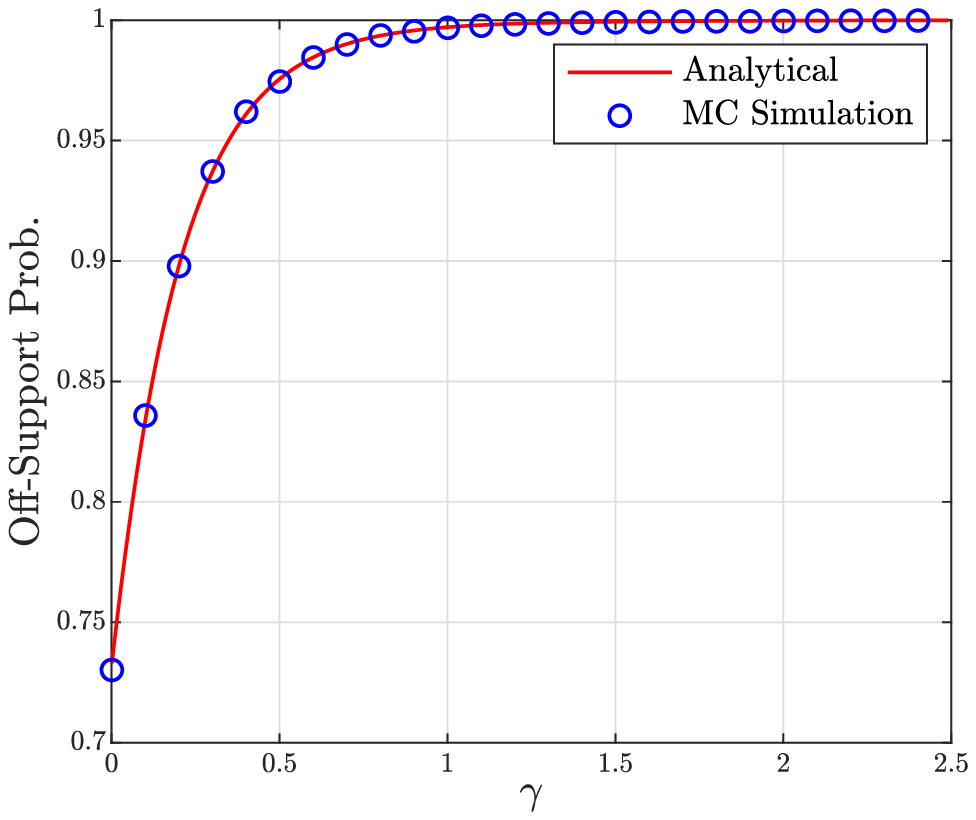}
\caption{\scriptsize{Off-Support Probability.}}
\end{subfigure}
\caption{\scriptsize{Probability of successful support recovery of the Box-LASSO. 
The analytical prediction is based on Theorem \ref{EN_on/off} with $p_{s_0} = (1-\kappa) \ \delta_0 + \kappa \ \delta_E$. We used $\kappa =0.2,\eta = 1.5,n = 128$, $T = 500,T_t =n, \nu = 0.5, E=1,$ and $P= 15 \ {\rm{dB}}$.}}%
\label{fig:on/off}
\end{figure*}
\begin{remark}\normalfont
For the previously mentioned sparse-Bernoulli distributed signal, 
the support recovery probabilities in (\ref{eq:on_supp}), and (\ref{eq:off_supp}) simplify to  
\begin{align}\label{eq:on-Bernoulli}
\underset{{n\to\infty} }{\rm{plim}} \ \Uppsi_{\upzeta}^{\text{on}}(\widehat{\sv}) = \ &Q \left( (\upzeta - E) \lambda_\star \sqrt{\eta P_d \sigma_{\widehat H}^2} +  \frac{\gamma \sqrt{P_d}}{\beta_\star  \sigma_{\widehat H}}   \right),
\end{align}
and 
\begin{equation}\label{eq:off-Bernoulli}
\underset{{n\to\infty} }{\rm{plim}} \ \Uppsi_{\upzeta}^{\text{off}}(\widehat{\sv}) = \ 1 - Q \left(  \upzeta  \lambda_\star \sqrt{\eta P_d \sigma_{\widehat H}^2}  + \frac{\gamma \sqrt{P_d}}{\beta_\star  \sigma_{\widehat H}}   \right).
\end{equation}
These expressions have been used in the numerical simulations above with $E =1$ therein.
\end{remark}
\begin{remark}[Unbounded elements]\normalfont
%
In instances when the elements of the original signal are \emph{unbounded} but take values in a specific range with high probability, the Box-LASSO can be a valuable decoder as well. 
To demonstrate this, let us take the example in which the elements of $\sv_0$ are i.i.d. sparse-Gaussian distributed, i.e., $S_0 \sim (1-\kappa)\delta_0 +\kappa \ \mathcal{N} (0, \sigma_s^2)$. 
\figref{fig:MSE_Gauss} illustrates a case in which the Box-LASSO outperforms the standard LASSO.
We used $\ell = - \sigma_s$ and $\mu = \sigma_s$ as the box-constraints in this example.
\end{remark}
\begin{figure}
  \centering
\includegraphics[width=8cm, height =6.1cm]{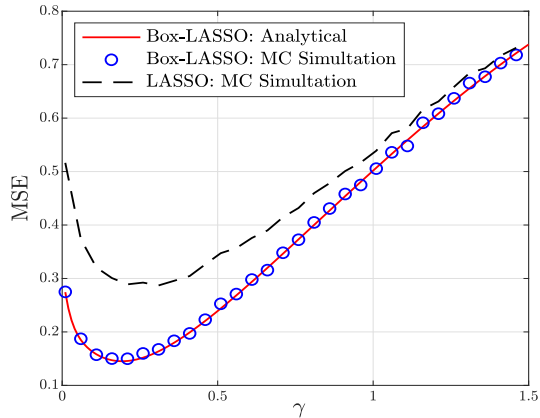}
\caption{\scriptsize{MSE of the Box-LASSO and standard LASSO. The analytical prediction is based on Theorem \ref{EN_mse} with a sparse-Gaussian $\sv_0$ signal.  We used $\kappa =0.1,\eta = 1.2,n = 400$, $T = 1000,T_t =456, \nu =0.5,\sigma_s^2=1,$ and $ P= 10 \ {\rm{dB}}$.}}%
\label{fig:MSE_Gauss}
\end{figure}
\begin{remark}[{Universality}]\normalfont
%
%
Even without the Gaussianity assumption on the elements of the channel matrix $\Hm$, our extensive simulations strongly indicate that the statements of Theorem~\ref{EN_mse} and Theorem~\ref{EN_on/off} are still valid. This is especially helpful in MIMO applications where the channel matrix elements can be represented beyond the typical fading model (Gaussian), such as in the more involved fading models, e.g., Weibull and Nakagami \cite{simon2005digital}. The same asymptotic statements appear to hold regardless of whether the distribution of $\Hm$ is Gaussian, Binary, or Laplacian (as illustrated in \figref{fig:univ}). Rigorous proofs, known as \textit{universality results}, in \cite{abbasi2019universality, oymak2018universality, montanari2017universality} support such a claim.
\end{remark}
\begin{figure}
  \centering
\includegraphics[width=7.3cm, height =6.1cm]{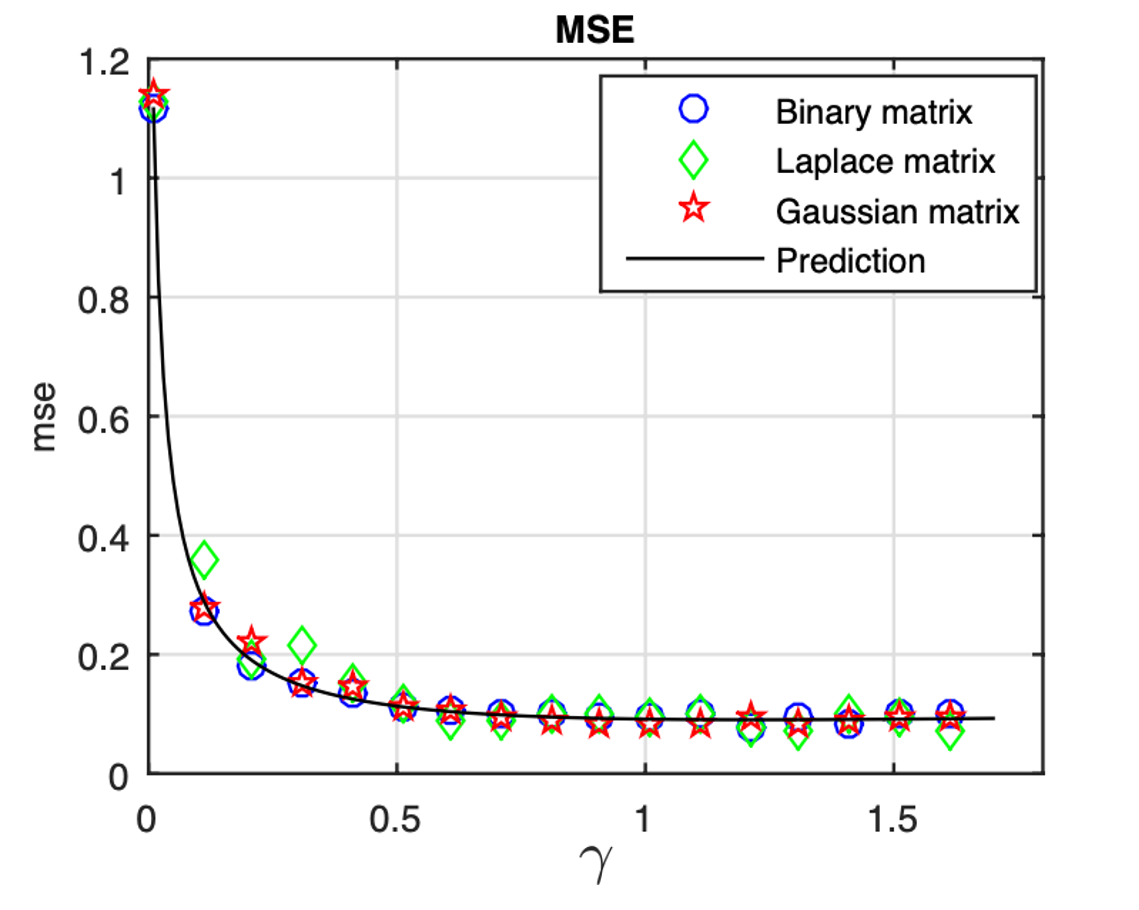}
\caption{\scriptsize{MSE of the Box-LASSO for different measurement matrices. The theoretical prediction is based on Theorem \ref{EN_mse} with a sparse-Bernoulli $\sv_0$ signal.  We used $\kappa =0.1,\eta = 0.8,n = 200$, $T = 700,T_t =256, \nu =0.6, P= 5 \ {\mathrm{dB}}$.}}%
\label{fig:univ}
\end{figure}
\begin{remark}[Efficient implementation]\normalfont
It is worth noting that the Box-LASSO can be efficiently implemented via \textit{quadratic programming (QP)} as in \cite{gaines2018algorithms}, where it was used to implement an efficient algorithm of the constrained LASSO. We applied the same algorithm to the Box-LASSO decoder in the above numerical simulations utilizing MATLAB built-in function $\mathsf{quadprog}$.
\end{remark}
\section{GSSK Modulated Massive MIMO Systems}
\label{sec:GSSK}
Traditional linear modulation schemes become expensive in modern massive MIMO systems. This is due to the large required number of radio frequency (RF) chains needed for the massive number of antennas. 
One promising modulation technique is the so-called spatial modulation (SM), where only the antenna's location relays information and only a small subset of the antennas is active at each time \cite{mesleh2008spatial}. This significantly reduces the system complexity since the required number of RF chains is less now. This modulation scheme saves energy; since using fewer RF chains, we have less power dissipation through the power amplifiers, etc. Modern SM techniques include generalized space-shift keying (GSSK) modulation \cite{jeganathan2008generalized, jeganathan2009space}, and the generalized spatial modulation (GSM) \cite{younis2010generalised}.
\subsection{Box-LASSO for Detecting GSSK Modulated Signals}
As mentioned above, recently developed modulation techniques such as GSM and GSSK modulation, generate signals that are essentially sparse and have elements belonging to a finite alphabet (i.e., bounded). 
Hence, when such modulations are employed, we will use the Box-LASSO as low-complexity decoding method, instead of the previously proposed standard LASSO decoders \cite{yu2012compressed, liu2013denoising}.

For the sake of simplicity, we will focus on GSSK.
Considering a modulation setup, where a fixed-size set $\mathcal{I} \subset \{1, \cdots, n\}$ of active antennas transmit $s_{0,j} = 1, j \in \mathcal{I}$ at each transmission, while the remaining antennas stay inactive, i.e., $s_{0,j} = 0, j \notin \mathcal{I}$.
Hence, only active antennas positions convey information.
%
%
%

To decode $\sv_0$, firstly, get a solution $\widehat{\sv}$ of the Box-LASSO in \eqref{Box-EN_1}, with $\ell =0$ and $\mu =1$. 
Then, map $\widehat{\sv}$ into a vector ${\sv}^\star$ which has elements either $0$ or $1$. 
In the GSSK context, this typically entails sorting $\widehat{\sv}$ and setting its largest $k$ entries to $1$ and the remaining to $0$ \cite{yu2012compressed}.
%

In order to evaluate the performance of the Box-LASSO in this application, we will use the so-called \emph{element error rate} (EER) \cite{atitallah2017box} as a performance measure. 
Similar to the support recovery metric, we first hard-thresholding $\widehat{\sv}$ by a constant $\upzeta \in (0, 1)$, in order to map each of its element to either $0$ or $1$.
Then, the EER can be defined as 
\begin{equation}\label{eer_def}
 {\rm{EER}}_\upzeta := \frac{1}{| \mathcal I|} \sum_{j \in \mathcal I} \mathbbm{1}_{\{| \widehat{s}_{j}| \leq \upzeta \}} + \frac{1}{n-| \mathcal I|} \sum_{j \notin \mathcal I} \mathbbm{1}_{\{| \widehat{s}_{j}| \geq \upzeta \}}.
\end{equation}

The next proposition gives a sharp asymptotic prediction of the EER in the GSSK modulated MIMO systems application.

\begin{proposition}\label{GSSK_EER} \normalfont
Let ${\rm{EER}}_\upzeta$ be as defined in the above equation with $|\mathcal I| = \kappa n$, for $\kappa \in (0, 1)$. Also, let $\beta_\star$,  and $\lambda_\star$ be solutions to the minimax optimization in \eqref{eq:scalar1}, with $p_{s_0} =  (1-\kappa) \delta_0 +\kappa \delta_{1} $ therein. Then, for a fixed $\upzeta \in (0,1)$, it holds
\begin{align}\label{GSSK_EER:eq}
\underset{{n \to \infty}}{\rm{plim}} \  {\rm{EER}}_{\upzeta} = \ & Q \left(  (1 -\upzeta ) \lambda_\star \sqrt{\eta P_d \sigma_{\widehat H}^2} - \frac{\gamma \sqrt{P_d}}{\beta_\star  \sigma_{\widehat H}} \right) 
+ Q \left(  \upzeta  \lambda_\star \sqrt{\eta P_d \sigma_{\widehat H}^2} +  \frac{\gamma \sqrt{P_d}}{\beta_\star  \sigma_{\widehat H}}   \right).
\end{align}
\end{proposition}
\begin{proof}
The proof follows from Theorem~\ref{EN_on/off} by observing that the EER in \eqref{eer_def} may be rewritten as
 \begin{equation}
 {\rm{EER}}_\upzeta =  2 -\Uppsi_{\upzeta}^{\text{on}}(\widehat{\sv}) -\Uppsi_{\upzeta}^{\text{off}}(\widehat{\sv}),
\end{equation}
with the on/off support probability expressions of the sparse-Bernoulli distribution derived earlier in \eqref{eq:on-Bernoulli} and \eqref{eq:off-Bernoulli} for $E=1$.
\end{proof}
On the receiving side of some MIMO systems, there may not be sufficient number of antennas. This is owing to the receiver's small size, limited cost or weight, and low power consumption.
Such MIMO systems, where the number of receive antennas $m$ is less than that of the transmitters $n$ (i.e., $m < n$), are known as \emph{overloaded} (or underdetermined) MIMO systems \cite{wong2007efficient}.
\figref{fig:EER} illustrates the accuracy of the derived EER expression for a case of an overloaded system, with $\eta =0.8$. This figure shows that the Box-LASSO outperforms the standard one in the EER sense as well. 
\begin{figure}[h!]%
\begin{center}
\includegraphics[width=8.7cm, height =6.50cm]{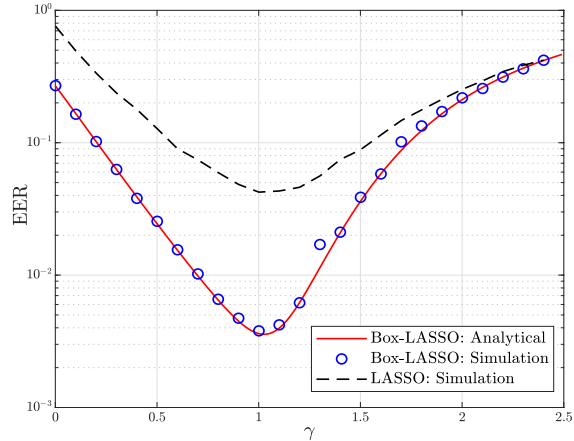}
\caption{\scriptsize{Element Error Rate of the Box-LASSO and the standard LASSO for GSSK signal recovery. Here, we used $T =500, m=120, n=T_t =150, k =15, \upzeta = 0.1, \nu = 0.5$, and $P = \mathrm{  10 \ dB}$. The data are averaged over $100$ independent iterations.}}%
\label{fig:EER}
\end{center}
\end{figure}

\subsection{Power Allocation and Training Duration Optimization}
The overall massive MIMO system's performance can be improved by optimizing the power allocation between the transmitted pilot and data symbols as oppose to equal power distribution \cite{hassibi2003much}.
Power optimization problems in MIMO systems have been proposed based on different performance metrics. 
In \cite{ballal2019optimum, zhao2017game}, the authors derived a power allocation scheme based on minimizing the MSE, while minimizing the the bit error rate (BER) and SER was considered in \cite{wang2014ber, alrashdi2020optimum, alrashdi2018optimum}. 
Training optimization based on maximizing the channel capacity was addressed in \cite{gottumukkala2009optimal, kannu2005capacity, hassibi2003much}.  In addition, the authors in \cite{dao2018pilot, muharar2020optimal, lu2018training} provided power allocation strategies based on maximizing the sum rates.
The above list of references is not comprehensive, since power allocation optimization research has very rich literature. However, we cited the most related works to this paper.

\subsubsection{Optimal Power Allocation}
In this subsection, we will use the previously derived asymptotic results for the MSE and EER to find an optimal power allocation scheme, in a GSSK modulated system, that minimizes these error measures.

For fixed $\tau_t $ and $\tau$, the power allocation optimization problem can be caste as
\begin{align*}
&\min_{P_t, P_d} \rm{MSE} \\
& \text{subject to:} \ P_t \tau_t  + P_d (\tau-\tau_t ) = {P} \tau, \\
& P_t = (1-\nu) {P} \tau,  \ P_d  = \nu {P} \tau, \ \ 0<\nu<1.
\end{align*}
It can be shown that the above optimization problem boils down to only optimizing the data energy ratio $\nu$.
The results are summarized next.

For fixed $\tau_t $ and $\tau$, and using Box-LASSO with an optimal regularizer $\gamma_\star^{\mathrm{MSE}}$ as in \eqref{opt_Reg_MSE}, the optimal power allocation that minimizes the MSE is given by
\begin{align}
\nu_\star^{\rm{MSE}} = \argmin_{0<\nu<1} \rm{MSE},
\end{align}
where $\rm {MSE}$ is the asymptotic MSE expression in \eqref{eq:MSE}. 

Similarly,  when using the EER as a performance metric, we have 
\begin{align}
\nu_\star^{\rm{EER}} = \argmin_{0<\nu<1} \rm{EER}_\upzeta(\gamma_\star),
\end{align}
where $\rm {EER}_\upzeta(\gamma_\star)$ is the asymptotic EER expression in \eqref{GSSK_EER:eq} with the optimal regularizer $\gamma_\star$ that minimizes the EER. 
For the Box-LASSO decoder, finding $\nu_\star^{\rm{MSE}}$ or $\nu_\star^{\rm{EER}}$ in closed-form expressions seems to be a difficult task, but by using a bisection method we can numerically find the optimal power allocation scheme. 

In \figref{fig:opt_power}, we plotted the MSE and EER of the Box-LASSO and standard LASSO as functions of the data energy ratio $\nu$.
This figure shows that optimizing the MSE and EER are equivalent with $\nu_\star \approx0.5373$. Furthermore, it shows that the optimal power allocation is nothing but the well-known scheme $\bar\nu_\star$ which was shown in \cite{alrashdi2020optimum, hassibi2003much} to maximize the effective SNR, where $\bar\nu_\star$ is given as
(\cite[eq. (35)]{alrashdi2020optimum}):
\begin{equation}\label{optimal_power}
\bar\nu_\star =
\begin{cases}
\vartheta  - \sqrt{\vartheta (\vartheta  -1)} & \text{if $\tau_{d} > 1$,} \\
\frac{1}{2} & \text{if $\tau_{d} =1$,} \\
\vartheta  + \sqrt{\vartheta (\vartheta  -1)} & \text{if $\tau_d< 1$,}
\end{cases}
\end{equation}
where 
\begin{align}
\vartheta = \frac{1 + P\cdot \tau}{P\cdot \tau (1 - \frac{1}{\tau_d})}.
\end{align}

This result is significant, since it showed again that the optimal power allocation scheme is nothing but the celebrated one that maximizes the effective SNR of the MIMO system, i.e., $\bar\nu_\star $. The power allocation actually does not depend on the type of the modulation constellation used, the used detector, or the other problem parameters such as $\eta$ and $\kappa$.\footnote{Provided that we use the LMMSE estimator for the channel estimation.}
For example, in \cite{alrashdi2020optimum}, the same power allocation scheme was obtained for a massive MIMO system with $M$-PAM signals and a Box-regularized least squares (Box-RLS) detector, while this work employs the Box-LASSO decoder for GSSK signal recovery.

\begin{figure*}
\begin{subfigure}{.5\textwidth}
  \centering
\includegraphics[width=8cm, height =6.1cm]{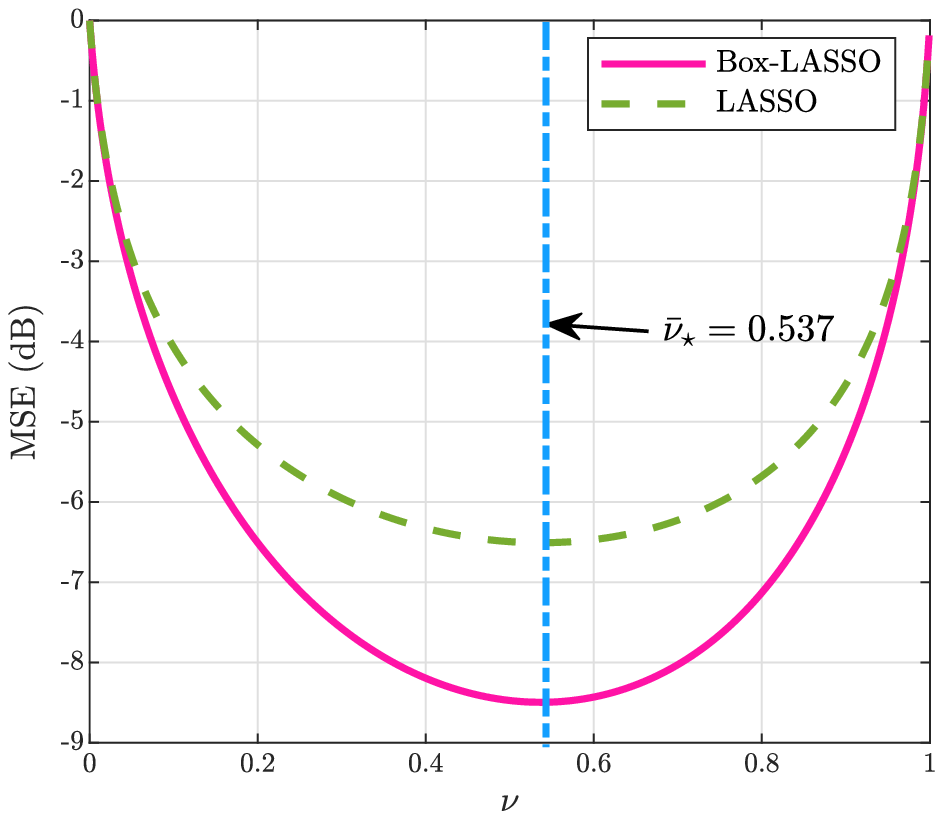}
\caption{\scriptsize{MSE vs. data energy ratio $\nu$.}}
\end{subfigure}
\begin{subfigure}{.5\textwidth}
\centering
\includegraphics[width=8cm, height =6.1cm]{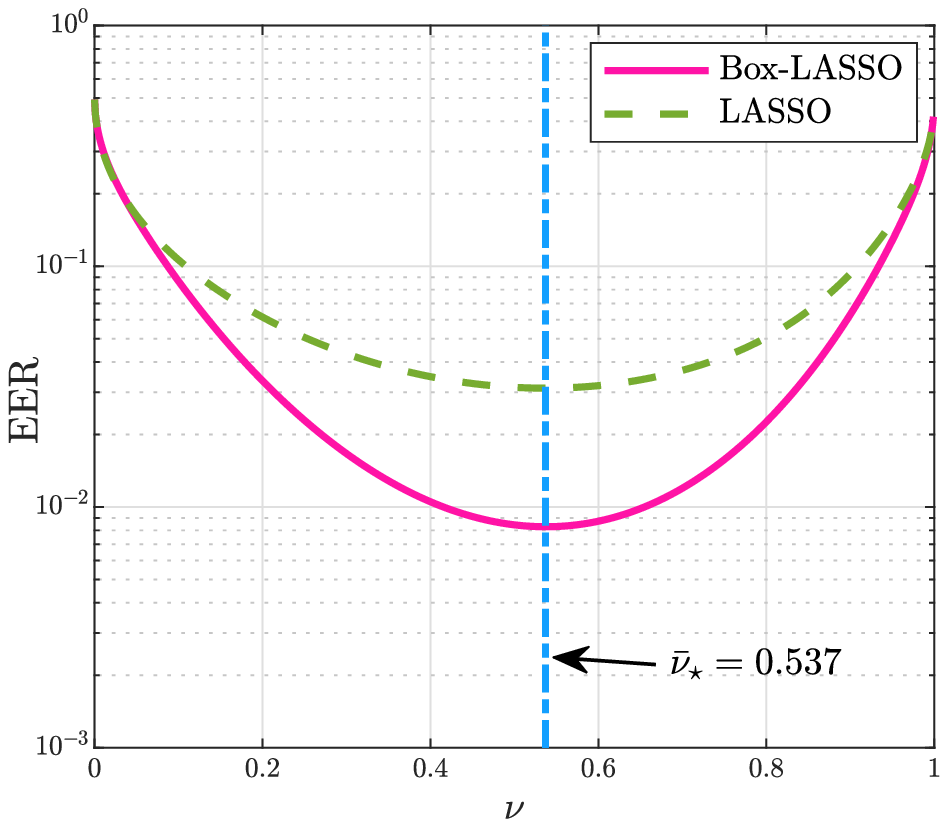}
\caption{\scriptsize{EER vs. data energy ratio $\nu$.}}
\end{subfigure}
\caption{\scriptsize{This figure shows the MSE/EER of Box-LASSO/LASSO as functions of $\nu$. We used $T =1000, n=400, T_t =456, P= 12 \ \mathrm{dB}, \eta =1.5,\kappa =0.1,$ and $ \upzeta =0.01$. }}%
\label{fig:opt_power}
\end{figure*}
\subsubsection{Optimal Training Duration}
In order to optimize the training duration, we introduce the following performance metric, which is called the \textit{goodput}. 
The goodput is calculated by dividing the amount of useful transmitted data by the time it takes to send it successfully \cite{grote2008ieee, alrashdi2020optimum}. Formally, it can be defined as 
\begin{align}\label{goodput:eq1}
{\rm{Goodput}}(\tau_t, \tau) = \bigg(1-\frac{\tau_t}{\tau}\bigg) (1-\rm{EER}).
\end{align}
The optimal value $T_t^\star$ that maximizes the goodput is determined in the following Corollary. 
For a fixed power allocation, the goal is to identify the optimal number of training symbols out of the total coherence interval symbols. 
From \eqref{eq:n pilots}, we must have $T_t^\star \geq n$ (or, $\tau_t^\star \geq 1$), 
and obviously,  $T_t^\star<T$ (or, $\tau_t^\star < \tau$). 
Therefore, $\tau_t^\star$ is a solution to the maximization problem:
\begin{align}\label{G:eq}
\tau_t^\star = \mathrm{arg} \max_{1 \leq \tau_t < \tau} \overline{\rm{Goodput}}(\tau_t, \tau),
\end{align}
where 
\begin{align}
\overline{\rm{Goodput}}(\tau_t, \tau): =\big(1-\frac{\tau_t}{\tau}\big) [1-{\rm{plim}}_{n \to \infty} \ \rm{EER}_\upzeta(\gamma_\star)]
\end{align}
 is the asymptotic value of the goodput.
\begin{corollary}[Optimal Training Duration]\label{Corr:opt_Training}\normalfont
Under imperfect CSI, the optimal training duration that maximizes the goodput in \eqref{G:eq} is given by:
\begin{align}
\tau_t^\star = 1 \ (\text{or} \ T_t^\star = n),
\end{align}
for all ${P}$ and $\tau$ (or $T$).
\end{corollary}
\begin{proof}
This result can be proven in a similar manner to \cite{alrashdi2020optimum}, details are thus omitted.
\end{proof}
\figref{fig:goodput} shows the goodput performance of Box-LASSO simulated versus the training duration $T_t$ which confirms the result of Corollary \ref{Corr:opt_Training}. It shows that at $T_t = n=200$, the goodput is maximized.
\begin{figure}[h]%
\begin{center}
\includegraphics[width=8.43cm, height =6.4cm]{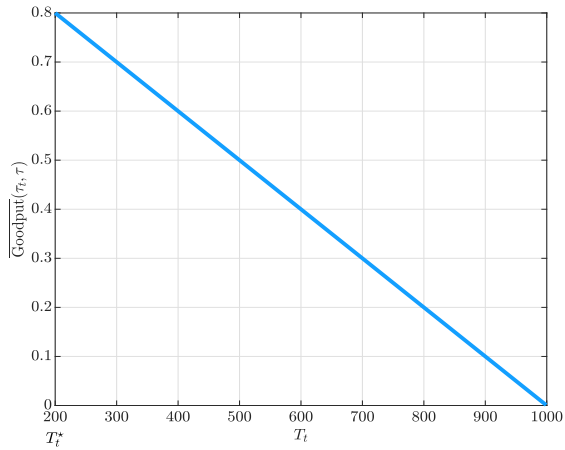}
\caption{\scriptsize{Goodput performance vs. $T_t$. We used $T = 1000, P=12 \ \mathrm{dB},n =200,\nu =0.5,\eta =1.5,\kappa =0.1, $ and $\upzeta =0.01$.}}%
\label{fig:goodput}
\end{center}
\end{figure}
\section{Conclusion and Future Work}
\label{sec:conc}
In this work, we derived sharp asymptotic characterizations of the mean square error, probability of support recovery and element error rate of the Box-LASSO under the presence of uncertainties in the channel matrix in the form of channel estimation errors.  
The analytical tool used in the analysis is the recently developed CGMT framework.
The derived expressions can be used to optimally tune the involved hyper-parameters of the algorithm such as the regularizer. 
Then, we used the Box-LASSO as a low complexity decoder in an application of massive MIMO detection using GSSK modulation, and optimize the power allocation between training and data symbols to minimize the MSE or EER of the system.
Furthermore, we derived the optimal training duration that maximizes the system's goodput. Numerical simulations show very close agreement to the derived theoretical predictions. Moreover, we showed that the Box-LASSO outperforms standard one in all of the considered performance metrics.

Finally, we highlight that the generalized spacial modulation (GSM) is more involved than the considered GSSK modulation since it uses the positions of active antenna in addition to a constellation symbol (e.g., $M$-QAM, $M$-PSK, etc.) to encode information \cite{younis2010generalised}. 
However, we focused in this paper on GSSK systems since the analysis framework, namely the CGMT, requires real-valued quantities (signals and channels). 
An interesting possible future work is to extend the results of this paper to systems involving complex-valued data such as GSM and investigate their power allocation optimization. 
Moreover, this paper assumes Gaussian channels matrices with i.i.d.  entries, but in numerous wireless communication applications, there are usually spatial correlations between the antennas. Therefore, another possible extension is to study correlated massive MIMO systems, where the matrix entries are no longer i.i.d.

\section*{Acknowledgments}
This research has been funded by the Scientific Research Deanship at the University
of Ha'il--Saudi Arabia through project number BA-2136.


\appendices
\appendix 
\section{Sketch of the Proof}
\label{sec:proof}
In this appendix, we derive the asymptotic analysis of the considered Box-LASSO problem's performance metrics. Our analysis is based on the CGMT, which is discussed more below.
\subsection{Technical Tool: CGMT}
%
Firstly, we summarize the CGMT framework \cite{thrampoulidis2016precise} before the proof of our main results.
For a comprehensive list of technical requirements, please see \cite{thrampoulidis2016precise, thrampoulidis2018symbol}.
Consider the following two optimization problems, which we call the Primal Optimization (PO) and Auxiliary Optimization (AO) problems, respectively.
\begin{subequations}
\begin{align}\label{P,AO}
&\Upphi(\Cm) := \underset{\xv \in \mathcal{S}_{\xv}}{\operatorname{\min}}  \ \underset{\yv \in \mathcal{S}_{\yv}}{\operatorname{\max}} \ \yv^{T} \Cm \xv + \xi( \xv, \yv), \\
&\upphi(\gv_1, \gv_2) := \underset{\xv \in \mathcal{S}_{\xv}}{\operatorname{\min}}  \ \underset{\yv \in \mathcal{S}_{\yv}}{\operatorname{\max}} \ \| \xv \| \gv_1^{T} \yv + \| \yv \| \gv_2^{T} \xv + \xi( \xv, \yv), \label{AA2}
\end{align}
\end{subequations}
where $\Cm \in \mathbb{R}^{\tilde m \times \tilde n}, \gv_1 \in \mathbb{R}^{\tilde m}, \gv_2 \in \mathbb{R}^{\tilde n}, \mathcal{S}_\xv \subset \mathbb{R}^{\tilde n}, \mathcal{S}_\yv \subset \mathbb{R}^{\tilde m}$ and $\xi: \mathbb{R}^{\tilde n} \times \mathbb{R}^{\tilde m} \mapsto \mathbb{R}$. Moreover, the function $\xi$ is assumed to be independent of the matrix $\Cm$. Denote by $\xv_{\Upphi} := \xv_{\Upphi}(\Cm) $, and $\xv_{\upphi} := \xv_{\upphi}( \gv_1, \gv_2)$ any optimal minimizers of (\ref{P,AO}) and (\ref{AA2}), respectively. Further let $\mathcal{S}_\xv,$ and $\mathcal{S}_\yv$ be convex and compact sets, $\xi(\xv,\yv)$ is convex-concave continuous on $\mathcal{S}_\xv \times \mathcal{S}_\yv$, and let $\Cm, \gv_1$ and $\gv_2 $ all have i.i.d. standard normal entries.

The PO-AO equivalence is formally stated in the next theorem, the proof of which can be found in \cite{thrampoulidis2016precise}.
\begin{theorem}[CGMT \cite{thrampoulidis2016precise}]\label{Th:CGMT}\normalfont
Under the above assumptions, the CGMT \cite[Theorem 3]{thrampoulidis2016precise} shows that the following
statements are true:\\
\noindent
\textbf{(a)} For all $c \in \mathbb{R}$ and $t>0$, it holds
\begin{align}
\mathbb{P}\left[\left| \Upphi(\Cm)-c \right|> t\right] \leq 2 \mathbb{P}\left[\left| \upphi(\gv_1, \gv_2)-c \right|> t\right].
\end{align}
In words, concentration of the optimal cost of the AO problem around $c$ implies concentration
of the optimal cost of the corresponding PO problem around the same value $c$. 
According the CGMT,  this will finally imply that the \emph{original} optimal cost (e.g., the Box-LASSO of \eqref{Box-EN_1} in this paper) will concentrate around $c$ as well. Please see \figref{fig:obj} for a numerical illustration.\\

\noindent
\textbf{(b)} Let $\mathcal{S}$ be any arbitrary open subset of $\mathcal{S}_\xv $, and $\mathcal{S}^c = \mathcal{S}_\xv \setminus\mathcal{S}$. Denote $\upphi_{\mathcal{S}^c}(\gv_1,\gv_2)$ the optimal cost of the optimization in (\ref{AA2}), when the minimization over $\xv$ is constrained over $\xv \in \mathcal{S}^c$. Consider the regime $\tilde m,  \tilde n \to\infty$ such that $\frac{\tilde m}{ \tilde n}\to \eta$. To keep notation short, this regime is denoted by $\tilde n\to\infty$. 
Suppose that there exist constants $\overline{\phi}$ and $\overline{\phi}_{\mathcal{S}^c}$ such that\\
(i) $\overline{\phi} < \overline{\phi}_{\mathcal{S}^c}$,
(ii) $\upphi(\gv_1,\gv_2) \overset{P}{\longrightarrow} \overline{\phi}$,  and
(iii) $\upphi_{\mathcal{S}^c}(\gv_1,\gv_2) \overset{P}{\longrightarrow} \overline{\phi}_{\mathcal{S}^c}$. 

\noindent
Then,  if in addition, $\lim_{\tilde n \rightarrow \infty} \mathbb{P}[\xv_{\upphi} \in \mathcal{S}] = 1,$ it also holds that 
\begin{align}
\lim_{\tilde n \rightarrow \infty} \mathbb{P}[\xv_{\Upphi} \in \mathcal{S}] = 1.
\end{align}
\end{theorem}
%
When the assumptions of Theorem~\ref{Th:CGMT} are met, the CGMT-based proof proceeds in general as follows:\\
$\bullet$ \textbf{Identifying the PO and the associated AO}: This step involves transforming the original optimization problem into the desired minimax PO form, and then identify its corresponding AO problem. \\
$\bullet$ \textbf{Simplifying the AO}: In this step, the AO is reduced into a scalar optimization problem.\\
$\bullet$ \textbf{Probabilistic analysis of the AO}: In this step, we prove that the AO converges to a (deterministic) asymptotic optimization problem which involves only scalar variables.\\
$\bullet$ \textbf{Choice of $\mathcal{S}$}: The set $\mathcal{S}$ should be selected properly based on the measure of interest. For instance, if we want to analyze the MSE or the EER, $\mathcal{S}$ will be the set in which the MSE or the EER concentrates, respectively.

After introducing the CGMT framework, we prove Theorem~\ref{EN_mse} next. For clarity, the steps of the proof are divided into different subsections.
\subsection{PO and AO Identification}
To obtain the result of the main theorems using CGMT, we need first to rewrite the Box-LASSO optimization problem \eqref{Box-EN_1} as a PO problem.
For convenience, we consider the \textit{error vector} 
\begin{align}
\ev := \sv - \sv_0, 
\end{align}
and the modified Box-set for all $j\in \{1,2, \cdots, n\}$:
\begin{equation}
\mathbb{D} =\bigg \{ e_j \in \mathbb{R} : \ell - s_{0,j} \leq e_j \leq \mu - s_{0,j} \bigg\},
\end{equation}
then the problem in (\ref{Box-EN_1}) can be reformulated as
\begin{align}\label{Lasso_w}
\widehat{\ev} = {\rm{arg}} \ \underset{\ev \in \mathbb{D}^n}{\operatorname{\min}} \ & \bigg\| \sqrt{ \frac{P_d}{n}} \widehat{\Hm} \ev + \sqrt{ \frac{P_d}{n}} \Omegam \sv_0 -\vv \bigg\|_2^2 +  \gamma P_d \| \ev + \sv_0 \|_1.
\end{align}\noindent
This minimization is not in the PO form as it is missing the max part. So to fix this, let us express the \emph{loss} function using its Legendre-Fenchel transformation\footnote{For any convex function $f$, we can write: $f(\xv) = \max_{\yv} \yv^T \xv - f^\star(\yv)$, where $f^\star$ is the Fenchel (convex) conjugate of $f$.}:
\begin{equation}
\| \xv \|_2^2 = \max_{\yv \in \mathbb{R}^m}  \yv^T \xv - \frac{1}{4} \| \yv\|_2^2.
\end{equation}
Hence, the problem above is equivalent to the following
\begin{align}\label{LASSO-PO1}
\underset{\ev \in \mathbb{D}^n}{\operatorname{\min}} \ \underset{\yv \in \mathbb{R}^m}{\operatorname{\max}} & \ \sqrt{ \frac{P_d}{n}}  \yv^T  \widehat{\Hm} \ev + \sqrt{ \frac{P_d}{n}}   \yv^T  \Omegam \sv_0 - \yv^T \vv - \frac{1}{4} \| \yv \|_2^2+  \gamma P_d \| \ev + \sv_0 \|_1.
\end{align}
One technical requirement of the CGMT is the compactness of the feasibility set over $\yv$. This can be handled following the approach in \cite[Appendix A]{thrampoulidis2016precise}, by introducing a sufficiently large {artificial} constraint set 
\begin{align}
\mathcal{S}_\yv = \biggl\{ \yv \in \mathbb{R}^m: \| \yv \|_2 \leq C_y \biggr \},
\end{align}
for some sufficiently large constant (independent of $n$) $C_y>0$. This will not asymptotically affect the optimization problem. Then, we obtain
\begin{align}\label{LASSO-PO22}
\underset{\ev \in \mathbb{D}^n}{\operatorname{\min}} \ \underset{\yv \in \mathcal{S}_\yv}{\operatorname{\max}}& \ \sqrt{ \frac{P_d \sigma_{\widehat{H}}^2}{n}}  \yv^T  \widetilde{\Hm} \ev + \sqrt{ \frac{P_d \sigma_{\omega}^2}{n}}   \yv^T  \widetilde\Omegam \sv_0 - \yv^T \vv - \frac{1}{4} \| \yv \|_2^2 +  \gamma P_d \| \ev + \sv_0 \|_1,
\end{align}
where $\widetilde{\Hm}$ and $\widetilde{\Omegam}$ are independent matrices with i.i.d. $\mathcal{N}(0,1)$ entries each.
The above problem is now in a \textit{PO form} with 
\begin{align}
\xi(\ev,\yv) = \ &\sqrt{ \frac{P_d \sigma_{\omega}^2}{n}}   \yv^T  \widetilde\Omegam \sv_0- \yv^T \vv    - \frac{1}{4} \| \yv \|_2^2+  \gamma P_d \| \ev + \sv_0 \|_1.
\end{align}
Thus,  its \textit{associated AO} is given as
\begin{align}\label{LASSO-AO1}
\underset{\ev \in \mathbb{D}^n}{\operatorname{\min}} \ \underset{\yv \in \mathcal{S}_\yv}{\operatorname{\max}} & \ \sqrt{ \frac{P_d \sigma_{\widehat{H}}^2}{n}}  \|\ev\|_2 \gv^T\yv + \sqrt{ \frac{P_d \sigma_{\widehat{H}}^2}{n}}  \|\yv\|_2 \zv^T\ev + \xi(\ev,\yv),
\end{align}
where $\gv \sim \mathcal{N}(\mathbf{0},\Id_m)$ and $\zv \sim \mathcal{N}(\mathbf{0},\Id_n)$ are independent random vectors.
\subsection{AO Simplification} 
In order to simplify the AO, we first let 
$$
\widetilde{\gv}:= \sqrt{ \frac{P_d \sigma_{\widehat{H}}^2}{n}}  \|\ev\|_2 \gv - \vv + \sqrt{ \frac{P_d \sigma_{\omega}^2}{n}}   \widetilde\Omegam \sv_0.
$$
It can be shown that $\widetilde{\gv} \sim \mathcal{N}(\mathbf{0}, \Sigmam_{\widetilde{\gv}}),$
where 
\begin{align}
\Sigmam_{\widetilde{\gv}} = \left(   \frac{P_d \sigma_{\widehat{H}}^2}{n} \| \ev\|_2^2 +1 +\frac{P_d \sigma_{\omega}^2}{n} \| \sv_0 \|_2^2    \right) \Id_m.
\end{align}
Then, the AO can be rewritten as
\begin{align}\label{LASSO-AO2}
\underset{\ev \in \mathbb{D}^n}{\operatorname{\min}} \ \underset{\yv \in \mathcal{S}_\yv}{\operatorname{\max}} \ & \widetilde\gv^T\yv + \sqrt{ \frac{P_d \sigma_{\widehat{H}}^2}{n}}  \|\yv\|_2 \zv^T\ev -\frac{\| \yv\|_2^2}{4} + \gamma P_d \| \ev + \sv_0 \|_1.
\end{align}

Fixing the norm of $\yv$ to $\alpha: =\| \yv \|_2$, we can easily optimize over its direction. This simplifies the AO to
\begin{align}\label{LASSO-AO22}
\underset{\ev \in \mathbb{D}^n}{\operatorname{\min}} \ \max_{\alpha \geq 0} & \ \alpha \|\widetilde\gv\|_2 + \sqrt{ \frac{P_d \sigma_{\widehat{H}}^2}{n}} \ \alpha \ \zv^T\ev  -\frac{\alpha^2}{4} + \gamma P_d \| \ev + \sv_0 \|_1.
\end{align}
\subsection{Asymptotic Analysis of the AO}
Next, we need to normalize the above optimization problem by $\frac{1}{n}$, to have all of its terms of the same order, $\mathcal{O}(1)$, and also define
\begin{align}
\beta: = \frac{\alpha}{\sqrt n}.
\end{align}
Then, after a change of the order of the $\min$-$\max$\footnote{\cite{thrampoulidis2016precise} shows that flipping the $\min$-$\max$ order is possible for large dimensions.}, we obtain:
\begin{align}\label{LASSO-AO23}
\max_{\beta \geq 0}  \underset{\ev \in \mathbb{D}^n}{\operatorname{\min}} & \ \beta \sqrt{\frac{P_d \sigma_{\widehat{H}}^2}{n} \| \ev\|_2^2 +1 +\frac{P_d \sigma_{\omega}^2}{n} \| \sv_0 \|_2^2} \ \ \frac{\|\gv\|_2}{\sqrt n} \nonumber \\
&+ \sqrt{ {P_d \sigma_{\widehat{H}}^2}}  \beta \frac{1}{n} \zv^T\ev -\frac{\beta^2}{4} + \frac{\gamma P_d}{n} \| \ev + \sv_0 \|_1.
\end{align}
Note the abuse of notation for $\gv$ to represent another standard normal vector.

To have a separable optimization problem, we use the following variational identity:
\begin{align}
\sqrt{x} = \min_{\lambda>0}  \frac{1}{2 \lambda} + \frac{ \lambda x}{2 },\ \text{for} \ x>0,
\end{align}
with optimum solution $\hat \lambda = \frac{1}{\sqrt x}$.
Using this trick,  with 
\begin{align}\label{eq:pre-MSE}
x = \frac{P_d \sigma_{\widehat{H}}^2}{n} \| \ev\|_2^2 +1 +\frac{P_d \sigma_{\omega}^2}{n} \| \sv_0 \|_2^2,
\end{align}
the optimization in (\ref{LASSO-AO23}) becomes
\begin{align}\label{AA24}
\max_{\beta \geq 0}  \min_{\lambda>0} & \ \frac{\beta \| \gv \|_2}{2 \lambda \sqrt{n}}+ \frac{\beta \lambda \| \gv \|_2}{2 \sqrt{n}} \left( 1+ \frac{P_d  \sigma_{\omega}^2}{n} \| \sv_0 \|_2^2 \right) -\frac{\beta^2}{4}  \nonumber \\ 
&+ \min_{\ev \in \mathbb{D}^n} \biggl\{   \frac{\beta \lambda \| \gv \|_2}{2 \sqrt n} \frac{P_d  \sigma_{\widehat{H}}^2}{n} \| \ev \|_2^2 + \sqrt{ {P_d \sigma_{\widehat{H}}^2}}  \beta \frac{\zv^T\ev }{n}+ \frac{\gamma P_d}{n} \| \ev + \sv_0 \|_1 \biggr\}.
\end{align}
Using the weak law of large numbers (WLLN), 
\begin{align}
 \frac{\| \gv\|_2}{\sqrt{n}} \pto \sqrt \eta,
\end{align}
and
\begin{align}
\frac{1}{n} \| \sv_0 \|_2^2 \pto \kappa \sigma_s^2=\kappa.
\end{align}
Next, using the above convergence results, and working with the original optimization variable $\sv$ instead of $\ev$, we get
\begin{align}\label{AA25}
\max_{\beta \geq 0}  \min_{\lambda>0} & \ \frac{\beta \sqrt{\eta}}{2 \lambda }+ \frac{\beta \lambda \sqrt{\eta}}{2 } \left( 1+ {P_d  \sigma_{\omega}^2} \kappa   \right) -\frac{\beta^2}{4}  \nonumber \\ 
&+\frac{1}{n}\sum_{j=1}^n \min_{\ell \leq s_j \leq \mu } \biggl\{  \frac{\beta \lambda \sqrt{\eta}}{2 } {P_d  \sigma_{\widehat{H}}^2} (s_j - s_{0,j})^2+ \sqrt{ {P_d \sigma_{\widehat{H}}^2}}  \beta z_j (s_j -s_{0,j})  + {\gamma P_d} | s_j | \biggr\}.
\end{align}
After a completion of squares in $s_j$, and noting that $\frac{1}{n} \zv^T \sv_0 \pto 0$, the above problem becomes
\begin{align}\label{AA26}
\max_{\beta \geq 0}  \min_{\lambda>0} & \ \frac{\beta \sqrt{\eta}}{2 \lambda }+ \frac{\beta \lambda \sqrt{\eta}}{2 } \left( 1+ {P_d  \sigma_{\omega}^2} \kappa   \right) -\frac{\beta^2}{4}  - \frac{1}{n} \sum_{j=1}^n \frac{\beta}{2 \lambda \sqrt{\eta}} z_j^2  \nonumber \\ 
&+\beta \lambda \sqrt{\eta} P_d  \sigma_{\widehat{H}}^2 \ \frac{1}{n}\sum_{j=1}^n \min_{\ell \leq s_j \leq \mu } \Biggl\{  \frac{1}{2} \Biggl( s_j- \biggl(  s_{0,j} + \frac{z_j}{\lambda \sqrt{\eta P_d  \sigma_{\widehat{H}}^2 }}\biggr) \Biggr)^2  + \frac{\gamma}{\beta \lambda \sqrt{\eta} \sigma_{\widehat{H}}^2} | s_j | \Biggr\} \nonumber\\
= & \max_{\beta \geq 0}  \min_{\lambda>0}  G(\beta,\lambda,\zv,\sv_0).
\end{align}
The optimization over $s_j$ could be obtained in closed-form using the \textit{saturated} shrinkage function $\mathcal{H}(a ; \gamma, l, u)$ which is defined in \eqref{soft_TH}. 
Also, let its optimal objective be $\mathcal{J}(a;\gamma,l,u) =\min_{l \leq s \leq u} \frac{1}{2}(s-a)^2 + \gamma |s|$ as defined in \eqref{optimal_Threshold}.
Now, again, using the WLLN, we have 
\begin{align}
\frac{1}{n} \sum_{j=1}^n z_j^2 \pto 1, 
\end{align}
and for all $\beta>0$, and $\lambda>0$:
\begin{align}
&\frac{1}{n} \sum_{j =1}^n \mathcal{J} \left(s_{0,j}+\frac{z_j}{\lambda \sqrt{\eta P_d  \sigma_{\widehat{H}}^2 }}; \frac{\gamma}{\beta \lambda \sqrt{\eta}  \sigma_{\widehat{H}}^2 }, \ell,\mu \right) 
\pto \mathbb{E} \left[\mathcal{J} \left(S_{0}+\frac{Z}{\lambda \sqrt{\eta P_d  \sigma_{\widehat{H}}^2 }}; \frac{\gamma}{\beta \lambda \sqrt{\eta}  \sigma_{\widehat{H}}^2 }, \ell,\mu \right)   \right],
\end{align}
where the expectation is taken over $S_0 \sim p_{s_0}$ and $Z \sim \mathcal{N}(0,1)$.\\
Consequently, the objective function in \eqref{AA26}, i.e., $G(\beta,\lambda,\zv,\sv_0)$, converges point-wise to the quantity $\mathcal{G}(\beta,\lambda)$, defined in \eqref{eq:scalar1}, in the limit of $n \to \infty$.
Afterwards, observe that $G(\beta,\lambda,\zv,\sv_0)$ is convex in $\lambda$ and concave in $\beta$. With these, and using Theorem 2.7 in \cite{newey1994large}, it follows that 
\begin{align}
\max_{\beta \geq 0}  \min_{\lambda>0}  G(\beta,\lambda,\zv,\sv_0) \pto \max_{\beta \geq 0}  \min_{\lambda>0} \mathcal{G}(\beta,\lambda).
\end{align}
Finally, the optimization problem in \eqref{AA26} simplifies to the following \emph{Scalar Optimization} (SO) problem:
\begin{align}\label{AA27}
\max_{\beta \geq 0}  \min_{\lambda>0} & \ \frac{\beta \sqrt{\eta}}{2 \lambda }+ \frac{\beta \lambda \sqrt{\eta}}{2 } \left( 1+ {P_d  \sigma_{\omega}^2} \kappa   \right) -\frac{\beta^2}{4}  - \frac{\beta}{2 \lambda \sqrt{\eta}}  \nonumber \\ 
&+\beta \lambda \sqrt{\eta} P_d  \sigma_{\widehat{H}}^2 \  \mathbb{E} \Biggl[\mathcal{J} \biggl(S_{0}+ \frac{Z}{\lambda \sqrt{\eta P_d  \sigma_{\widehat{H}}^2 }}; \frac{\gamma}{\beta \lambda \sqrt{\eta}  \sigma_{\widehat{H}}^2 }, \ell,\mu \biggr)   \Biggr] \nonumber \\
&= \max_{\beta \geq 0}  \min_{\lambda>0} \mathcal{G}(\beta,\lambda).
\end{align}
It worth mentioning that in the above equation, $ \mathcal{G}(\beta_\star,\lambda_\star)$, where $(\beta_\star,\lambda_\star)$ is the unique solution of \eqref{AA27}, represents the 
the asymptotic value of the objective function in \eqref{Box-EN_1} for a minimizer $\widehat{\sv}$, i.e.,
\begin{align}
\underset{n\to\infty}{\rm{plim}} \ \frac{1}{n}\left( \bigg\|   \sqrt{ \frac{P_d}{n}} \widehat{\Hm} \widehat{\sv} - \rv  \bigg\|_2^2 +  \gamma P_d  \| \widehat{\sv}  \|_1 \right) = \mathcal{G}(\beta_\star,\lambda_\star).
\end{align}
\figref{fig:obj} shows the great accuracy of the above result.
\begin{figure}
  \centering
\includegraphics[width=8.5cm, height =6.5cm]{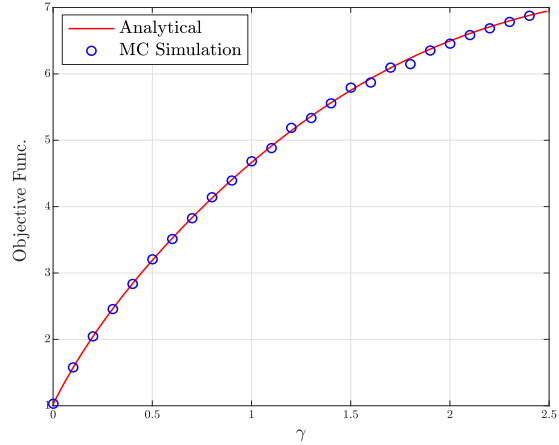}
\caption{\scriptsize{Optimal objective function value of the Box-LASSO vs. the regularizer for a sparse-Bernoulli vector. We used $\kappa =0.2,\eta = 1.5,n = 128$, $T = 500,T_t =n, \nu = 0.5, E=1,$ and $P= 15 \ {\rm{dB}}$.}}%
\label{fig:obj}
\end{figure}

After deriving the SO problem, we are now in a position to study the asymptotic convergence of the MSE. The analysis is given in the next subsection.

\subsection{Error Analysis via CGMT (Proof of Theorem~1)}
In this part, we study the asymptotic convergence of the MSE of the Box-LASSO. First, using the fact that $\hat \lambda = \frac{1}{\sqrt x}$ and recalling from (\ref{eq:pre-MSE}) that 
\begin{align}
 \frac{1}{n} \| \widetilde \ev\|_2^2= \frac{1}{P_d \sigma_{\widehat{H}}^2}  \left( \frac{1}{\hat\lambda_n^2}-1 -\frac{P_d \sigma_{\omega}^2}{n} \| \sv_0 \|_2^2 \right),
\end{align}
where $\widetilde \ev$ is the AO solution in \eqref{LASSO-AO23}, and $\hat\lambda_n$ is the solution to \eqref{AA24} in $\lambda$.
Using \cite[Lemma 10]{thrampoulidis2016precise}, it can be shown that $\hat \lambda_n \pto \lambda_\star$, where $\lambda_\star$ is the solution to \eqref{AA27}. Then, by the WLLN, $\frac{1}{n} \| \sv_0\|_2^2 \pto \kappa$, and hence
\begin{align}\label{M-eq}
\frac{1}{n} \| \widetilde{\ev} \|_2^2&\pto \frac{1}{P_d \sigma_{\widehat{H}}^2}  \left( \frac{1}{\lambda_\star^2}-1 -{P_d \sigma_{\omega}^2} \kappa   \right).
\end{align}
Recall that $\widetilde{\ev} =\widetilde{\sv} - \sv_0$,
so the last step is to use the CGMT to prove that the quantities $\widehat{\sv} - \sv_0$ and $\widetilde{\sv} - \sv_0$ are concentrated in the same set with high probability. Formally,
for any fixed $\varepsilon > 0$, we define the set: 
{{
\begin{align}
\mathcal{S}_\varepsilon = \left\{ \mathbf{q} \in \mathbb{R}^n: \left|  \frac{\| \qv \|_2^2}{n}  - \frac{ (\frac{1}{\lambda_\star^2}-1 -{P_d \sigma_{\omega}^2} \kappa  )   }{P_d \sigma_{\widehat{H}}^2}    \right| < \varepsilon \right\}.
\end{align}
}}
Equation (\ref{M-eq}) proves that for any $\varepsilon>0$, $\widetilde{\sv} - \sv_0 \in \mathcal{S}_\varepsilon$ with probability approaching one. 
Then, we conclude by applying the CGMT that $\widehat{\sv} - \sv_0 \in \mathcal{S}_\varepsilon$ with probability approaching one.
This proves the asymptotic prediction of the MSE as summarized in Theorem \ref{EN_mse}.

The \emph{residual} prediction in Remark~\ref{Remark:Res}, eq. \eqref{eq:res}, can be proven in a similar way, by first noting that 
\begin{align}
\|\hat\yv \|_2^2 = 4n \mathcal{R},
\end{align}
where $\hat\yv$ is the PO solution of \eqref{LASSO-PO1}. Using the definition: $\beta_n^2 = \frac{\|\widetilde \yv\|_2^2}{n}$, where $\beta_n, $ and $\widetilde \yv$ are the solutions of \eqref{LASSO-AO2}, and \eqref{LASSO-AO23} respectively, and following the same steps as in the MSE proof above, one can show that $\hat\yv$ and $\widetilde\yv$ concentrate in the same set with probability approaching one, and then apply the CGMT to reach the proof of the residual result in \eqref{eq:res}. Details are thus omitted.
\bibliographystyle{IEEEbib}
\bibliography{References02}
\end{document}